\documentclass[runningheads]{llncs}
\usepackage{graphicx}
\usepackage{xspace}
\usepackage{tikz}
\usetikzlibrary{arrows, calc, shapes}
\usepackage{extarrows}
\usepackage{paralist}
\usepackage{amsfonts}
\usepackage{textcomp}
\usepackage[noend]{algpseudocode}

\usepackage{hyperref}

\newcommand{\m}[1]{\mathsf{#1}}
\newcommand{\mc}[1]{\mathcal{#1}}
\renewcommand{\vec}[1]{\overline{#1}}
\renewcommand{\phi}{\varphi}
\renewcommand{\theta}{\vartheta}

\newcommand{\seqz}[2][n]{{#2_0},\dots,{#2_{#1}}}
\newcommand{\Var}{\mathcal Var}
\newcommand{\Const}{L}
\newcommand{\tool}{\texttt{ada}\xspace}

\renewcommand{\AA}{\mc A} 
\newcommand{\BB}{\mc B} 
\newcommand{\CC}{\mc C} 
\newcommand{\NN}{\mc N} 
\newcommand{\KK}{\mc K} 
\newcommand{\HH}{\mc H} 
\newcommand{\NFAfor}[1][\psi]{{\NN}_{#1}}
\newcommand{\inn}{\,{\in}\,} 
\newcommand{\eqn}{\,{=}\,} 
\newcommand{\modelsLTL}{\models_{\KK}}
\newcommand{\qd}{\mathit{qd}}

\newcommand{\NFApsi}[1][\psi]{{\NN}_{#1}}
\newcommand{\last}{\lambda} 
\newcommand{\inquotes}[1]{\ensuremath{\text{\textgravedbl}\!{#1}\!\text{\textacutedbl}}}
\newcommand{\domino}[4]{\text{\tikz[baseline=-0.5ex]{\node[scale=.6, inner sep=0pt]{$%
\left[\begin{array}{@{\,}l@{=}l@{\,}}{#1}&{#2}\\{#3}&{#4}\\\end{array}\right]$}}}}

\newcommand{\ints}{\mathit{int}}
\newcommand{\rats}{\mathit{rat}}
\newcommand{\reals}{\mathit{real}}

\newcommand{\MC}{\textup{MC}} 

\newcommand{\E}{\mathsf{E}\,}
\newcommand{\A}{\mathsf{A}\,}
\newcommand{\X}{\mathsf{X}\,}
\newcommand{\Xw}{\mathsf{Y}\,}
\newcommand{\G}{\mathsf{G}\,}
\newcommand{\F}{\mathsf{F}\,}
\newcommand{\U}{\mathrel{\mathsf{U}}}
\newcommand{\CTLsf}{CTL$^*_f$\xspace}

\newcommand{\mydds}{DDSA\xspace}

\newcommand{\LTLconf}[1][\CC]{\smash{\textup{LTL}_f^\BB}\xspace}
\newcommand{\CTL}{CTL$^*_f$\xspace }

\newcommand{\binit}{b_{\mathit{I}}} 
\newcommand{\bdummy}[1][b]{\underline{#1}} 
\newcommand{\alphainit}{\alpha_{\mathit{I}}} 
\newcommand{\goto}[1]{\mathrel{\raisebox{-2pt}{$\xrightarrow{#1}$}}}

\newcommand{\FRuns}{\mathit{FRuns}}
\newcommand{\trans}[1]{\Delta_{#1}} 
\newcommand{\update}{\mathit{update}} 
\newcommand{\guard}{\mathit{guard}} 
\newcommand{\vseq}[1][\theta]{\overline{#1}}

\newcommand{\hist}{h} 
\newcommand{\tonode}[1][{}]{\mathit{to}} 
\newcommand{\formulas}[1][V]{{\CC(#1)}} 
\newcommand{\leqn}{\,{\leq}\,} 
\newcommand{\geqn}{\,{\geq}\,} 
\newcommand{\pcnode}[3]{{#1}\nodepart{two}{#2}\nodepart{three}{#3}}

\newcommand{\Confs}{\mathcal{K}_\BB} 
\newcommand{\Conf}{K} 
\newcommand{\chS}{\mathit{checkState}} 
\newcommand{\chP}{\mathit{checkPath}} 
\newcommand{\toLTLconf}{\mathit{toLTL}_\KK} 

\newcommand{\lemref}[1]{Lem.~\ref{lem:#1}}

\newcommand{\defref}[1]{Def.~\ref{def:#1}}
\newcommand{\defsref}[2]{Defs.~\ref{def:#1} and~\ref{def:#2}}
\newcommand{\corref}[1]{Cor.~\ref{cor:#1}}

\newcommand{\secref}[1]{Sec.~\ref{sec:#1}}
\newcommand{\appref}[1]{App.~\ref{app:#1}}

\newcommand{\thmref}[1]{Thm.~\ref{thm:#1}}

\newcommand{\exaref}[1]{Ex.~\ref{exa:#1}}
\newcommand{\figref}[1]{Fig.~\ref{fig:#1}}

\renewcommand{\eqref}[1]{(\ref{eq:#1})}

\newtheorem{innernumberedlemma}{Lemma}
\newenvironment{numberedlemma}[1]
  {\renewcommand\theinnernumberedlemma{#1}\innernumberedlemma}
  {\endinnernumberedlemma}
\newtheorem{innernumberedtheorem}{Theorem}
\newenvironment{numberedtheorem}[1]
  {\renewcommand\theinnernumberedtheorem{#1}\innernumberedtheorem}
  {\endinnernumberedtheorem}

\newcommand{\casebox}[1]{\parbox{25mm}{\textbf{case} #1:}}
\algnewcommand\algorithmicswitch{\textbf{switch}}
\algnewcommand\algorithmiccase{\textbf{case}}
\algnewcommand\algorithmicassert{\texttt{assert}}
\algnewcommand\Assert[1]{\State \algorithmicassert(#1)}%
\algdef{SE}[SWITCH]{Switch}{EndSwitch}[1]{\algorithmicswitch\ #1\ \algorithmicdo}{\algorithmicend\ \algorithmicswitch}%
\algdef{SE}[CASE]{Case}{EndCase}[1]{\algorithmiccase\ #1}{\algorithmicend\ \algorithmiccase}%
\algtext*{EndSwitch}%
\algtext*{EndCase}%

\tikzstyle{place}=[draw, circle, inner sep=1.5pt, line width=.7pt, scale=.6, minimum width=5mm]
\tikzstyle{trans}=[draw, rectangle, inner sep=1.5pt, line width=.7pt, scale=.6, minimum width=4mm, minimum height=4mm]
\tikzstyle{state}=[draw, circle, inner sep=1.5pt, line width=.7pt, scale=.6]
\tikzstyle{edge}=[draw, ->, line width=.5pt]
\tikzstyle{action}=[scale=.6]
\tikzstyle{caption}=[scale=.9]

\begin{document}
\title{CTL$^*$ model checking for data-aware dynamic systems with arithmetic%
\thanks{This work is partially supported by the UNIBZ projects DaCoMan, QUEST, SMART-APP, VERBA, and WineId.}
}
%
%
\author{Paolo Felli \and Marco Montali \and Sarah Winkler}
\authorrunning{P. Felli, M. Montali, S. Winkler}
%
\institute{Free University of Bolzano-Bozen, Bolzano, Italy
\email{\{pfelli,montali,winkler\}@inf.unibz.it}}
\maketitle              
\begin{abstract}
The analysis of complex dynamic systems is a core research topic in 
formal methods and AI, and combined modelling of systems with data has gained increasing importance
in applications such as business process management.
In addition, process mining techniques are nowadays used to automatically mine process models from event data, often without correctness guarantees.
Thus verification techniques for linear and
branching time properties are needed to ensure desired behavior.

Here we consider data-aware dynamic systems with arithmetic (DDSAs), which constitute a concise but expressive formalism of transition systems with linear arithmetic guards.
We present a CTL$^*$ model checking procedure for DDSAs that relies on 
a finite-state abstraction by means of a set of formulas that capture variable configurations.
Linear-time verification was shown to be decidable in specific classes of DDSAs where
the constraint language or the control flow are suitably confined. 
We investigate several of these restrictions for the case of CTL$^*$, with both positive and negative results:
CTL$^*$ verification is proven decidable for monotonicity and integer periodicity constraint 
systems, but undecidable for feedback free and bounded lookback systems.
To demonstrate the feasibility of our approach, we implemented it in the SMT-based prototype \tool, showing that many practical business process models can be effectively analyzed.
\end{abstract}

\section{Introduction}
\label{sec:intro}

The study of complex dynamic systems is a core research topic in AI, with a long tradition in formal methods. It finds application in a variety of domains, such as notably business process management (BPM), where studying the interplay between control-flow and data has gained momentum~\cite{Reichert12,CGM13,CDMP18,DHLV18}. Processes are increasingly mined by automatic techniques~\cite{Aalst16,BaralG15} 
that lack any correctness guarantees, making verification even more important to ensure the desired behavior.
However, the presence of data pushes verification to the verge of undecidability due to an infinite state space. 
This is aggravated by the use of arithmetic, in spite of its importance for practical applications \cite{DHLV18}.
Indeed, model checking of transition systems operating on numeric data variables with arithmetic constraints is known to be undecidable, as it is easy to model a two-counter machine.

In this work, we focus on the concise but expressive framework of data-aware dynamic systems with arithmetic (DDSAs) \cite{LFM20,ada}, also known as counter systems~\cite{IbarraS99,ComonJ98,DD07}.
Several classes of DDSAs have been isolated where specific verification tasks are decidable, notably reachability~\cite{IbarraS99,FinkelL02,ComonJ98,BGI09} and 
linear-time model checking~\cite{DD07,DemriG08,DDV12,LFM20,ada}.
Far fewer results are known about the case of branching time, with the exception of flat counter systems where loops may not be nested \cite{DemriFGD10}, and
gap-order constraint systems where constraints are restricted to the form $x\,{-}\,y\:{\geq}\:2$~\cite{BP14,MT16}.
However, many processes in BPM and beyond fall into neither of these two classes, as illustrated by the example below.

\begin{example}
\label{exa:road fines}
The following DDSA $\BB$ models a management process for road fines by the Italian police~\cite{MannhardtLRA16}. It maintains seven so-called \emph{case data} variables (i.e., variables local to each process instance, called ``case'' in the BPM literature):
$a$ (amount), $t$ (total amount), $d$ (dismissal code), $p$ (points deducted),
$e$ (expenses), and time durations $\mathit{ds}$, $\mathit{dp}$, $\mathit{dj}$.
The process starts by creating a case, upon which the offender is notified within 90 days, i.e., 2160h (\textsf{send fine}).
If the offender pays a sufficient amount $t$, the process terminates via silent actions $\tau_1$, $\tau_2$, or $\tau_3$.
For the less happy paths, the \textsf{credit collection} action is triggered 
if the payment was insufficient; while
\textsf{appeal to judge} and \textsf{appeal to prefecture} reflect filed protests by the offender, which again need to respect certain time constraints.
\\
%
\resizebox{\columnwidth}{!}
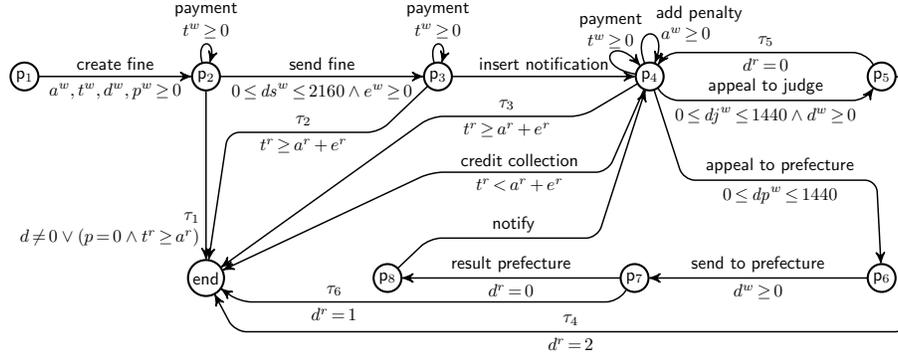
{
\centering
\begin{tikzpicture}[node distance=49mm,>=stealth', yscale=.95]
\node[state] (p1) {$\m p_1$}; 
\node[state, right of=p1, xshift=-13mm] (p2) {$\m p_2$}; 
\node[state, right of=p2, xshift=-3mm] (p3) {$\m p_3$}; 
\node[state, right of=p3, xshift=-7mm] (p4) {$\m p_4$}; 
\node[state, below of=p2, yshift=9mm] (end) {$\m {end}$}; 
\node[state, right of=p4, xshift=-3mm] (p5) {$\m p_5$}; 
\node[state, below of=p5,yshift=9mm] (p6) {$\m p_6$}; 
\node[state, left of=p6] (p7) {$\m p_7$}; 
\node[state, left of=p7] (p8) {$\m p_8$}; 
\draw[edge] (p1) -- 
  node[above, action] {\textsf{create fine}}
  node[below, action] {$a^w, t^w, d^w, p^w \geqn 0$}
  (p2);
\draw[edge] (p2) to[loop above, looseness=9]
  node[above, yshift=2.5mm,action] {\textsf{payment}}
  node[above, action] {$t^w \geqn 0$}
  (p2);
\draw[edge] (p2) -- 
  node[above, action] {\textsf{send fine}}
  node[below, action] {$0 \leqn \mathit{ds}^w \leqn 2160 \wedge e^w \geqn 0$}
  (p3);
\draw[edge] (p2) -- 
  node[left, action, near end] {\textsf{$\tau_1$}}
  node[left, action, near end, yshift=-4mm] {$d\,{\neq}\,0 \vee (p \eqn 0 \wedge t^r \geqn a^r)$}
  (end);
\draw[edge] (p3) to[loop above, looseness=9]
  node[above, yshift=2.5mm,action] {\textsf{payment}}
  node[above, action] {$t^w \geqn 0$}
  (p3);
\draw[edge] (p3) -- 
  node[above, action] {\textsf{insert notification}}
  (p4);
\draw[edge, rounded corners] (p3) -- ($(p3) + (-.7,-.7)$) --
  node[left, action, above] {\textsf{$\tau_2$}}
  node[left, action, below] {$t^r \geqn a^r\,{+}\,e^r$}
  ($(p3) + (-2.5,-.7)$) --(end);
\draw[edge] (p4) to[loop above, looseness=11, out=170, in=130]
  node[above, yshift=2.5mm,action] {\textsf{payment}}
  node[above, action] {$t^w \geqn 0$}
  (p4);
\draw[edge] (p4) to[loop right, looseness=11, out=110, in=70]
  node[right, yshift=2.5mm,action] {\textsf{add penalty}}
  node[right, action, xshift=1mm] {$a^w \geqn 0$}
  (p4);
\draw[edge, rounded corners] (p4) -- ($(p4) + (.2,-.3)$) --
  node[above, action] {\textsf{appeal to judge}}
  node[below, action] {$0 \leqn \mathit{dj}^w \leqn 1440 \wedge d^w \geqn 0$}
  ($(p5) + (-.2,-.3)$) -- (p5);
\draw[edge, rounded corners] (p4) -- ($(p4) + (-.5,-1.2)$) --
  node[above, action] {\textsf{credit collection}}
  node[below, action] {$t^r\,{<}\,a^r\,{+}\,e^r$}
  ($(p4) + (-2.6,-1.2)$) -- (end);
\draw[edge, rounded corners] (p4) -- ($(p4) + (-.8,-.5)$) --
  node[above, action] {\textsf{$\tau_3$}}
  node[below, action] {$t^r \geqn a^r\,{+}\,e^r$}
  ($(p4) + (-2.6,-.5)$) --(end);
\draw[edge, rounded corners] (p5) -- ($(p5) + (-.2,.3)$) --
  node[above, action] {\textsf{$\tau_5$}}
  node[below, action] {$d^r \eqn 0$}
  ($(p4) + (.2,.3)$) -- (p4);
\draw[edge, rounded corners] (p4) -- ($(p4) + (.4,-1.3)$) --
  node[above, action] {\textsf{appeal to prefecture}}
  node[below, action] {$0 \leqn \mathit{dp}^w \leqn 1440$}
  ($(p4) + (2.7,-1.3)$) -- (p6);
\draw[edge] (p6) to
  node[left, action, above] {\textsf{send to prefecture}}
  node[left, action, below] {$d^w \geqn 0$}
  (p7);
\draw[edge] (p7) to
  node[left, action, above] {\textsf{result prefecture}}
  node[left, action, below] {$d^r \eqn 0$}
  (p8);
\draw[edge, rounded corners] (p7) -- ($(p7) + (-.2,-.3)$) --
  node[left, action, above, near end] {\textsf{$\tau_6$}}
  node[left, action, below, near end] {$d^r \eqn 1$}
  ($(end) + (.4,-.3)$) -- (end);
\draw[edge, rounded corners] (p5) -- ($(p5) + (.3,0)$) -- ($(p5) + (.3,-3.2)$) --
  node[left, action, above] {\textsf{$\tau_4$}}
  node[left, action, below] {$d^r \eqn 2$}
  ($(p5) + (-7.7,-3.2)$)  -- (end);
\draw[edge, rounded corners] (p8) -- ($(p8) + (.5,.5)$) --
  node[left, action, above] {\textsf{notify}}
  ($(p8) + (2.5, .5)$) -- (p4);
\end{tikzpicture}
}
This model was generated from real-life logs
by automatic process mining techniques paired with domain knowledge \cite{MannhardtLRA16}, but without any correctness guarantee.
For instance, \emph{data-aware soundness}~\cite{FLM19,BatoulisHW17} requires that the process can always reach a final state from any reachable configuration, expressed by the branching-time property $\A\G \E\F \m{end}$.
This property is false here, as $\BB$ can get stuck in state
$\m p_7$ if $d\,{>}\,1$.
In addition, process-specific linear-time properties are needed, e.g., that a \textsf{send fine}
event is always followed by a sufficient payment (i.e., $\langle\mathsf{send\:fine}\rangle\top \to \F\langle\mathsf{payment}\rangle (t \geq a)$,
where $\langle \alpha \rangle$ is the next operator via action $\alpha$).
\end{example}

This example highlights how both linear-time and branching-time verification are needed. 
In this paper, we present a CTL$^*$ model checking algorithm for DDSAs, adopting a finite-trace semantics (\CTLsf)~\cite{MPRS20} to reflect the nature of processes as in \exaref{road fines}. More precisely, our approach can synthesize \emph{conditions} on the initial variable assignment such that a given property holds.
We then derive an abstract decidability criterion which is satisfied by
two practical DDSA classes that restrict the constraint language to
(a) monotonicity constraints~\cite{DD07,FLM19}, i.e., variable-to-variable or variable-to-constant comparisons over $\mathbb Q$ or $\mathbb R$, and
(b) integer periodicity constraints~\cite{DemriG08,Demri06}, i.e., variable-to-constant  and restricted variable-to-variable comparisons with modulo operators.
On the other hand, the restrictions known as \emph{feedback-freedom}~\cite{DDV12} and the more general \emph{bounded lookback}~\cite{ada} restrict the control flow of DDSAs such that
LTL$_f$ verification is decidable, but we show here that \CTLsf remains undecidable.

\noindent
In summary, we make the following contributions:
\begin{compactenum}
\item
We present a \CTLsf model checking algorithm for DDSAs; 
\item
As an abstract decidability criterion for our verification problem, we prove a termination condition for this algorithm (\corref{termination}); 
\item This result is used to show that \CTLsf verification is decidable for 
monotonicity constraint and integer periodicity constraint systems; 
\item The cases of feedback-free and bounded-lookback systems are undecidable; 
\item
We implemented our approach in the prototype \tool using SMT solvers as backends
and tested it on a range of business processes from the literature. 
\end{compactenum}
\smallskip

\noindent
The paper is structured as follows:
The rest of this section compiles related work.
In \secref{background} we recall preliminaries about DDSAs and CTL$^*_f$. \secref{ltl} is dedicated to LTL verification with \emph{configuration maps}, which is used by our model checking procedure in \secref{modelchecking}. After giving an abstract termination criterion, \secref{criteria}
presents decidability results for concrete DDSA classes.
We describe our implementation in \secref{implementation}. 
Complete proofs and experiments can be found in the appendix.

\paragraph{Related work.}
Verification of transition systems with arith\-metic constraints, also called counter systems,
has been studied in many areas including 
formal methods, database theory, and BPM.
Reachability was proven decidable for a variety of classes, e.g., reversal-bounded counter machines~\cite{IbarraS99}, finite linear~\cite{FinkelL02}, flat~\cite{ComonJ98}, and gap-order constraint (GC) systems~\cite{BGI09}.
Considerable work has also been dedicated to linear-time verification:
LTL model checking is decidable for monotonicity constraint (MC) systems~\cite{DD07}, even comparing variables multiple steps apart. 
DDSAs with MCs are also considered in \cite{FLM19} from the perspective of
 LTL with a finite-run semantics (LTL$_f$ \cite{dGV13}), giving an explicit procedure to compute finite, faithful abstractions.
Linear-time verification is also decidable for integer periodicity constraint systems, also with past time operators~\cite{Demri06,DemriG08}; and feedback-free systems, for an enriched constraint language that can refer to a read-only database~\cite{DDV12}.
Decidability of LTL$_f$ was also shown for systems with the 
abstract \emph{finite summary} property~\cite{ada}, which includes MC, GC, and systems with $k$-bounded lookback, the latter being a generalization of feedback freedom.

Branching-time verification was less studied:
Decidability of CTL$^*$ was proven for flat counter systems with Presburger-definable loop iteration~\cite{DemriFGD10}, even in NP~\cite{DemriDS18}. These results are orthogonal to ours:  we do not demand flatness, but our approach does not cover their results.
Moreover, it was shown that CTL$^*$ verification is decidable for pushdown systems, which can model counter systems with a single integer variable~\cite{FinkelWW97}.
For integer relational automata (IRA), i.e., systems with constraints $x\,{\geq}\,y$ or $x\,{>}\,y$ and domain $\mathbb Z$, CTL model checking is undecidable while the existential and universal fragments of CTL$^*$ remain decidable~\cite{Cerans94}.
For GC systems, which extend IRAs to constraints of the form $x-y \geq k$, the existential fragment of CTL$^*$ is decidable
while the universal one is not~\cite{BP14}.
A similar di\-cho\-tomy holds for the EF and EG fragments of CTL~\cite{MT16}.
A subclass of IRAs were considered in~\cite{CKL16,BG06}, allowing only periodicity and monotonicity constraints. While satisfiability of CTL$^*$ was proven decidable,
model checking is not (as already shown in~\cite{Cerans94}), though it is decidable for properties in the fragment CEF$^+$, an extension of the EF fragment~\cite{BG06}.
In contrast, rather than restricting temporal operators, we show decidability of model checking under an abstract property of the DDSA and the verified property.
This abstract property can be guaranteed by suitably constraining
the constraint class in the system, or the control flow.
More closely related is work by Gascon~\cite{Gascon09}, who shows decidability of CTL$^*$ model checking for counter systems that admit a \emph{nice symbolic valuation abstraction}, an abstract property which includes MC and integer periodicity constraint (IPC) systems.
The relationship between our decidability criterion and the property defined by Gascon will need further investigation.
Another difference is that we here adopt a finite-path semantics
for CTL$^*$ as e.g. considered in~\cite{SRM18}, since for the analysis of real-world processes such as business processes it is sufficient to consider finite traces.
%
On a high level, our method follows a common approach to CTL$^*$: the verification property is processed bottom-up,
and we compute solutions for each subproperty. These are then used to formulate an equivalent linear-time verification problem~\cite[p.429]{BK08}.
For the latter, we can partially rely on earlier work~\cite{ada}.

\section{Background}
\label{sec:background}

We start by defining the set of constraints over expressions of sort $\ints$, $\rats$, or $\reals$,
with associated domains $dom(\ints) = \mathbb Z$, $dom(\rats) = \mathbb Q$, and $dom(\reals) = \mathbb R$.
\begin{definition}
\label{def:constraint}
For a given set of sorted variables $V$, expressions $e_s$ of sort $s$ and atoms $a$ are defined as follows: \\
\begin{tabular}{r@{\,}l@{\quad}r@{\,}l}
$e_s$&:=\:$v_s\:\mid\:k_s\:\mid\:e_s\,{+}\,e_s\:\mid\:e_s\,{-}\,e_s\:$ &
$a$&:=\:$e_s\,{=}\,e_s\:\mid\:e_s\,{<}\,e_s\:\mid\:e_s\,{\leq}\,e_s\:\mid\:e_{\ints} \,{\equiv_n}\,e_{\ints}$ 
\end{tabular}\\
where $k_s\,{\in}\,dom(s)$, $v_s\,{\in}\,V$ has sort $s$, and $\equiv_n$ denotes equality modulo some $n\,{\in}\,\mathbb N$. A \emph{constraint} is then a quantifier-free boolean expression over atoms $a$.
\end{definition}
\noindent
The set of all constraints built from atoms over variables $V$ is denoted by $\CC(V)$.
For instance, $x \neq 1$, $x < y\,{-}\,z$, and $x\,{-}\,y = 2 \wedge y \neq 1$ are valid constraints
independent of the sort of $\{x, y, z\}$,
while $u \equiv_3 v + 1$ is a constraint for integer variables $u$ and $v$. 
We write $\Var(\phi)$ for the set of variables in a formula $\phi$.
For an assignment $\alpha$ with domain $V$ that maps variables to values in their domain,
and a formula $\phi$ we write $\alpha \models \phi$ if $\alpha$ satisfies $\phi$.

We are thus in the realm of SMT with linear arithmetic, which is decidable and admits \emph{quantifier elimination}~\cite{Presburger29}:
if $\phi$ is a formula in $\CC(X \cup \{y\})$,
thus having free variables $X \cup \{y\}$,
there is a quantifier-free $\phi'$ with free
variables $X$ that is equivalent to $\exists y. \phi$, 
i.e., $\phi'\,{\equiv}\,\exists y. \phi$,
where $\equiv$ denotes logical equivalence.
%

\subsection{Data-aware Dynamic Systems with Arithmetic}

From now on, $V$ will be a fixed, finite set of variables.
We consider two disjoint, marked copies of $V$,  $V^r = \{v^r \mid v\in V\}$ and $V^w = \{v^w \mid v\in V\}$, called the \emph{read} and \emph{write} variables.
They will refer to the variable values before and after a transition, respectively.
We also write $\vec V$ for a vector that contains the variables $V$ in an arbitrary but fixed order, and $\vec V^r$ and $\vec V^w$ for the vectors that order $V^r$ and $V^w$ in the same way.

\begin{definition}
\label{def:DDS}
A \emph{DDSA} $\BB=\langle B, \binit, \AA, T, B_F, V, \alphainit, \guard\rangle$ is
a labeled transition system where
\begin{inparaenum}[(i)]
\item $B$ is a finite set of \emph{control states}, with $\binit\inn B$ the initial one;
\item $\AA$ is a set of \emph{actions};
\item $T \subseteq B \times \AA \times B$ is a \emph{transition relation};
\item $B_F \subseteq B$ are \emph{final states};
\item $V$ is the set of \emph{process variables};
\item $\alphainit$ the \emph{initial variable assignment};
\item $\guard\colon \AA \mapsto \CC(V^r \cup V^w)$ specifies the \emph{executability constraints}. 
\end{inparaenum}
\end{definition}

\begin{example}
\label{exa:ddsas}
We consider the following DDSAs
$\BB$, $\BB_{\mathit{bl}}$, and $\BB_{\mathit{ipc}}$, where $x,y$ have domain $\mathbb Q$ and $u$, $v$, $s$ have domain $\mathbb Z$.
Initial and final states have incoming arrows and double borders, respectively;
$\alphainit$ is not fixed for now.\\[.5ex]
\resizebox{\textwidth}{!}
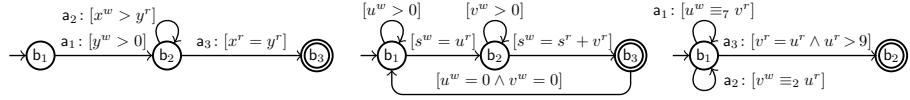
{
\begin{tikzpicture}[node distance=32mm]
\node[state] (1) {$\m b_1$};
\node[state, right of=1,xshift=-5mm] (2) {$\m b_2$};
\node[state, right of=2, double] (3) {$\m b_ 3$};
\draw[edge] ($(1) + (-.4,0)$) -- (1);
\draw[edge] (1) -- node[above,action]{$\m a_1\colon[y^w > 0]$} (2);
\draw[edge] (2) to[loop left, out=120,in=60, looseness=6] node[action, xshift=-1mm,yshift=1mm]{$\m a_2\colon[x^w > y^r]$} (2);
\draw[edge] (2) -- node[above,action]{$\m a_3\colon[x^r = y^r]$} (3);
\begin{scope}[shift={(45mm,-0mm)}, node distance=27mm]
\node[state] (1)  {$\m b_ 1$};
\node[state, right of=1, xshift=-5mm] (2) {$\m b_2$};
\node[state, right of=2, xshift=2mm, double] (3) {$\m b_3$};
\draw[edge] ($(1) + (-.4,0)$) -- (1);
\draw[edge] (1) to node[above,action] {$[s^w = u^r]$} (2);
\draw[edge] (2) to node[above,action]{$[s^w = s^r + v^r]$} (3);
\draw[edge, rounded corners] (3) -- ($(3) - (0,.5)$)
  -- node[above,action, pos=0.55]{$[u^w = 0 \wedge v^w = 0]$}  ($(1) - (0,.5)$)  -- (1);
\draw[edge] (1) to[loop above, out=120,in=60, looseness=6] node[action, yshift=-1mm]{$[u^w > 0]$} (1);
\draw[edge] (2) to[loop above, out=120,in=60, looseness=6] node[action, yshift=-1mm]{$[v^w > 0]$} (2);
\end{scope}
\begin{scope}[shift={(85mm,-0mm)}, node distance=40mm]
\node[state] (1)  {$\m b_1$};
\node[state, right of=1, double] (2) {$\m b_2$};
\draw[edge] ($(1) + (-.4,0)$) -- (1);
\draw[edge] (1) to node[above,action] {$\m a_3\colon [v^r \eqn u^r \wedge u^r\,{>}\,9]$} (2);
\draw[edge] (1) to[loop above, out=120,in=60, looseness=6] node[action, yshift=-1mm]{$\m a_1\colon [u^w \equiv_7 v^r]$} (1);
\draw[edge] (1) to[loop below, out=-120,in=-60, looseness=6] node[action,right, xshift=3mm, yshift=2mm]{$\m a_2\colon [v^w \equiv_2 u^r]$} (1);
\end{scope}
\end{tikzpicture}
}
\end{example}
Also the system in \exaref{road fines} represents a DDSA.
If state $b$ admits a transition to $b'$ via action $a$, namely $(b, a, b')\in \Delta$, this is denoted by $b \goto{a} b'$.
A \emph{configuration} of $\BB$ is a pair $(b, \alpha)$ where $b\inn B$
and $\alpha$ is an assignment with domain $V$.
A \emph{guard assignment} is a function $\beta\colon V^r \cup V^w \mapsto D$.
For an action $a$, let $write(a) = \Var(\guard(a)) \cap V^w$.
As defined next, an action $a$ transforms a configuration $(b, \alpha)$ into a new configuration $(b', \alpha')$ by updating the assignment $\alpha$ according to the action guard, which can at the same time evaluate
conditions on the current values of variables and write new values:

\begin{definition}
A \mydds $\BB\,{=}\,\langle B, \binit, \AA, T, B_F, V, \alphainit, \guard\rangle$
\emph{admits a step} from configuration $(b, \alpha)$ to 
$(b', \alpha')$ via action $a$,
denoted $(b, \alpha) \goto{a} (b', \alpha')$,
if $b \goto{a} b'$,
$\alpha'(v) = \alpha(v)$ for all $v \in V\setminus write(a)$, and
the guard assignment $\beta$ given by
$\beta(v^r) = \alpha(v)$ and
$\beta(v^w) = \alpha'(v)$ for all $v \in V$
satisfies $\beta \models \guard(a)$.
\end{definition}

\noindent
For instance, for $\BB$ in \exaref{ddsas} and initial assignment 
$\alphainit(x) = \alphainit(y) = 0$, the initial configuration admits a step
$\smash{(\m b_1, \domino{x}{0}{y}{0}) \goto{\m{a}_1}
(\m b_{2}, \domino{x}{0}{y}{3})}$ with $\beta(x^r) = \beta(x^w) = \beta(y^r) = 0$ and
$\beta(y^w) = 3$.

A \emph{run} $\rho$ of a DDSA $\BB$ of length $n$ from configuration $(b, \alpha)$ is a sequence of steps
$\rho\colon(b, \alpha) = (b_0, \alpha_0) 
\goto{a_1} (b_1, \alpha_1)
\goto{a_2} \dots
\goto{a_n} (b_n, \alpha_n)$.
We also associate with $\rho$ the \emph{symbolic run} 
$\smash{\sigma\colon b_0 \goto{a_1} b_1 \goto{a_2} \dots \goto{a_n} b_n}$
where state and action sequences are recorded without assignments, and say that 
$\sigma$ is the \emph{abstraction} of $\rho$ (or, $\sigma$ \emph{abstracts} $\rho$). For some $m<n$, 
$\sigma|_{m}$ denotes the prefix of $\sigma$ that has $m$ steps.

\subsection{History Constraints}

In this section, we fix a
\mydds $\BB = \langle B, \binit, \AA, T, B_F, V, \alphainit, guard\rangle$.
We aim to build an abstraction of $\BB$ that covers the (potentially infinite)
set of configurations by finitely many states of the form $(b,\varphi)$,
where $b \inn B$ is a control state and $\varphi$  a formula
that expresses conditions on the process variables $V$.
A state $(b,\varphi)$ will thus represent all configurations $(b,\alpha)$ s.t. $\alpha \models \varphi$.
To mimic steps on the abstract level, we define below the $\update$ function
to express how such a formula $\phi$ is modified by executing an action.
First, let the \emph{transition formula} of action $a$
be $\trans{a}(\vec V^r, \vec V^w)\:{=}\:
\guard(a) \wedge \bigwedge_{v\in V\setminus\mathit{write}(a)} v^{w}\,{=}\,v^{r}$. 
Intuitively, this states conditions on variables \emph{before and after} executing $a$: $\guard(a)$ must be true and the values of all variables that are not written are propagated by inertia. 
As $\trans{a}$ has free variables $\vec V^r$ and 
$\vec V^w$, we write $\trans{a}(\vec X, \vec Y)$ for the formula
obtained from $\trans a$ by replacing $\vec V^r$ by $\vec X$ and $\vec V^w$ by $\vec Y$.

\begin{definition}
\label{def:update}
For a formula $\phi$ with free variables $V$ and action $a$, 
$\update(\phi, a) = \exists \vec U. \phi(\vec U) \wedge \Delta_a(\vec U, \vec V)$,
where $U$ is a set of variables that do not occur in $\phi$.
\end{definition}

Our approach generates an abstraction using formulas of a special shape called \emph{history constraints} \cite{ada}, obtained by iterated $\update$ operations in combination with a sequence of \emph{verification constraints} $\vec \theta$.
The latter will later be taken from the transition labels of an automaton for the verified property.
For now it is enough to consider $\vseq$ an arbitrary sequence of constraints 
with free variables $V$. Its prefix of length $k$ is denoted by $\vseq|_k$. 
We need a fixed set of placeholder variables $V_0$ that are disjoint from $V$,
and assume an injective variable renaming $\nu\colon V\,{\mapsto}\,V_0$.
Let $\phi_\nu$ be the formula $\phi_\nu = \bigwedge_{v\in V} v \eqn \nu(v)$.

\begin{definition}\label{def:history constraint}
For a symbolic run $\sigma\colon b_0 \goto{a_1} b_1 \goto{a_2} \dots \goto{a_n} b_n$,
and verification constraint sequence $\vseq= \langle \theta_0,\dots, \theta_n\rangle$,
the \emph{history constraint} $\hist(\sigma, \vseq)$ 
is given by
$\hist(\sigma, \vseq) \eqn \phi_\nu \wedge \theta_0$ if $n\,{=}\,0$, and 
$\hist(\sigma, \vseq) \eqn \update(\hist(\sigma|_{n-1}, \vseq|_{n-1}), a_{n}) \wedge \theta_n$ if $n\,{>}\,0$.
\end{definition}

Thus, history constraints are formulas with free variables $V \cup V_0$.
Satisfying assignments for history constraints are closely related to assignments in runs:%
\footnote{
\lemref{abstraction} is a slight variation of \cite[Lem. 3.5]{ada}:
\defref{history constraint} differs from history constraints in~\cite{ada} in that the initial assignment is not fixed. We provide a proof
in \appref{proofs}.}

\begin{lemma}
\label{lem:abstraction}
For a symbolic run $\sigma\colon b_0 \goto{a_1} b_1 \goto{a_2} \dots \goto{a_n} b_n$ and $\vseq= \langle \theta_0,\dots, \theta_n\rangle$,
$\smash[t]{\hist(\sigma,\vseq)}$
is satisfied by assignment $\alpha$ with domain $V\,{\cup}\,V_0$ iff 
$\sigma$ abstracts a run
$\smash{\rho \colon (b_0, \alpha_0) \goto{a_1} \dots \goto{a_n} (b_n, \alpha_n)}$
such that 
\begin{inparaenum}[\it (i)]
\item $\alpha_0(v) = \alpha(\nu(v))$, and
\item $\alpha_n(v) = \alpha(v)$ for all $v\in V$, and
\item $\alpha_i \models \theta_i$ for all $i$, $0\leq i \leq n$.
\end{inparaenum}
\end{lemma}

\subsection{CTL$_f^*$}

For a DDSA $\BB$ as above,
we consider the following verification properties:
\begin{definition}
\label{def:ctlstar}
CTL$^*_f$ state formulas $\chi$ and path formulas $\psi$ are defined by the following grammar, for constraints $c \inn \CC(V)$ and control states $b\inn B$:
\begin{align*}
\chi &:= \top \mid c \mid b \mid 
\chi \wedge \chi \mid \neg \chi \mid  \E \psi &
\psi &:= \chi \mid \psi \wedge \psi \mid \neg \psi \mid \X \psi  \mid \G \psi \mid \psi \U \psi
\end{align*}
\end{definition}
We use the usual abbreviations $\F \psi = \top \U \psi$, $\chi_1 \vee\chi_2 = \neg(\neg \chi_1 \wedge \neg \chi_2)$, and $\A \psi = \neg \E \neg \psi$.
To simplify the presentation, we do not explicitly treat next state operators $\langle a\rangle$ via a specific action $a$, as used in \exaref{road fines}, 
though this would be possible (cf. \cite{ada}).
However, such an operator can be encoded by adding a fresh data variable $x$ to $V$,
the conjunct $x^w\eqn 1$ to $\guard(a)$, and $x^w\eqn 0$ to all other guards, and
replacing $\langle a\rangle \psi$ in the verification property
by $\X (\psi\wedge x=1)$.

The maximal number of nested path quantifiers in a formula $\psi$ is called the \emph{quantifier depth} of $\psi$, denoted by $\qd(\psi)$.
We adopt a finite path semantics for CTL$^*$~\cite{MPRS20}:
For a control state $b\in B$ and a state assignment $\alpha$, let $\FRuns(b,\alpha)$ be
the set of \emph{final runs}
$\smash{
\rho \colon (b,\alpha) = (b_0,\alpha_0) \goto{a_1} \dots \goto{a_n}
(b_n,\alpha_n)
}$
such that $b_n \in F$ is a final state.
The $i$-th configuration $(b_i, \alpha_i)$ in $\rho$ is denoted by $\rho_i$.

\begin{definition}
\label{def:semantics}
The semantics of CTL$^*_f$ is inductively defined as follows.
For a DDSA $\BB$ with configuration $(b,\alpha)$,
state formulas $\chi$,\:$\chi'$, and path formulas $\psi$,\:$\psi'$:

\noindent
\begin{tabular}{@{~}l@{ ~ }l}
$(b,\alpha) \models \top$ & \\
$(b,\alpha)\models c$  & iff $\alpha\models c$\\
$(b,\alpha) \models b'$ & iff $b = b'$\\
$(b,\alpha) \models \chi \wedge \chi'$ & iff 
$(b,\alpha)\models \chi$ and $(b,\alpha)\models \chi'$\\
$(b,\alpha) \models \neg \chi$ & iff $(b,\alpha) \not\models \chi$\\
$(b,\alpha) \models \E\psi$ & iff  $\exists \rho \in \FRuns(b,\alpha)$
such that $\rho \models \psi$
\end{tabular}
\smallskip

\noindent
where $\rho \models \psi$ iff $\rho,0 \models \psi$ holds, 
and for a run $\rho$ of length $n$ and all $i$, $0\,{\leq}\,i\,{\leq}\,n$:

\noindent
\begin{tabular}{@{ ~ }l@{ ~ }l}
$\rho,i \models \chi$ & iff $\rho_i \models \chi$ \\
$\rho, i \models \neg \psi$ &iff $\rho, i \not\models \psi$ \\
$\rho, i \models \psi \wedge \psi'$ &iff $\rho, i \models \psi$ and $\rho, i \models \psi'$ \\
$\rho,i \models \X \psi$ & iff $i < n$ and $\rho,i+1 \models \psi$\\
$\rho,i \models \G \psi$ & iff for all $j$, $i \leq j \leq n$, it holds that $\rho,j\models {\psi}$\\
$\rho,i\models {\psi}\U{\psi'}$ & iff $\exists k$ with $i+k \leq n$
such that $\rho,i+k\models \psi'$ \\
& and for all $j$, $0 \leq j < k$, it holds that $\rho,i+j\models {\psi}$.
\end{tabular}
\end{definition}

Instead of simply checking whether the initial configuration of a DDSA $\BB$ satisfies a \CTLsf property
$\chi$, we try to determine, for every state $b \in B$, which constraints on variables need to hold in order to satisfy $\chi$. As the number of configurations $(b,\alpha)$ of a DDSA $\BB$ is usually infinite, configuration sets cannot be enumerated explicitly. Instead, we represent a set of configurations as a \emph{configuration map} $K\colon B \mapsto \formulas$
that associates with every control state $b \in B$ a formula $K(b) \in \CC(V)$, 
representing all configurations $(b,\alpha)$ such that $\alpha \models K(b)$. 
Our aim is thus to compute a \emph{solution} $K$ to the following problem:

\begin{definition}[Verification problem]
\label{def:vp}
For a \mydds $\BB$ and state formula $\chi$, is there a configuration map $K$
such that $(b,\alpha) \models \chi$ iff $\alpha \models K(b)$, for all $b\inn B$?
\end{definition}

\noindent
We call the verification problem given by  $\BB$ and $\chi$ \emph{solvable} if a solution $K$ exists and can be effectively computed. 
For instance, for $\BB$ from \exaref{ddsas} and
$\chi_1 = \A \G (x\,{\geq}\,2)$, a solution is given by
$K = \{ \m b_1 \mapsto \bot,\ \m b_2 \mapsto x\,{\geq}\,2\wedge y\,{\geq}\,2,\ \m b_3 \mapsto x \geq 2\}$.
For $\chi_2 \models \E \X (\A \G (x \geq 2))$, a solution is
$K' = \{ \m b_1 \mapsto x\,{\geq}\,2,\ \m b_2 \mapsto y\,{\geq}\,2 ,\ \m b_3 \mapsto \bot\}$.
As $\m b_1$ is the initial state, $\BB$ satisfies $\chi_2$ with every initial assignment that sets $\alphainit(x) \geq 2$.
Note that a solution $K$ to the verification problem for $\BB$ and $\chi$ in particular allows to determine whether $(\binit,\alphainit) \models \chi$ holds, by testing $\alphainit \models K(\binit)$, so that $(\binit,\alphainit) \models \chi$ is decidable for $\BB$.

\section{LTL with Configuration Maps}
\label{sec:ltl}

Following a common approach to CTL$^*$ verification, our technique
processes the property $\chi$ bottom-up, computing solutions for each subformula $\E\psi$,
before solving a linear-time model checking problem $\chi'$ in which the solutions to subformulas appear as atoms.
Given our representation of sets of configurations,
we use LTL formulas where atoms are configuration maps, and
denote this specification language by $\LTLconf$.
For a given DDSA $\BB$, it is formally defined as follows:
\[
\psi\ :=\ 
K \mid
\psi \wedge \psi \mid \neg \psi \mid 
\X \psi \mid 
\G \psi \mid 
\psi \U \psi
\]
where $K\in \Confs$, for $\Confs$ is the set of configuration maps for $\BB$.
We again use a finite-trace semantics~\cite{dGV13}:
\begin{definition}
\label{def:ltl:semantics}
A run $\rho$ of length $n$ \emph{satisfies} an $\LTLconf$ formula $\psi$, denoted 
$\rho \modelsLTL \psi$, iff $\rho,0 \modelsLTL \psi$ holds, 
where for all $i$, $0 \leq i \leq n$:

\noindent
\begin{tabular}{@{}l@{~}l@{}}
$\rho,i \modelsLTL K$ & iff $\rho_i = (b,\alpha)$ and $\alpha\models K(b)$;\\
$\rho,i \modelsLTL \psi \wedge \psi'$ & iff 
$\rho,i \modelsLTL \psi$ and $\rho,i \modelsLTL \psi'$;\\
$\rho,i \modelsLTL \neg \psi$ & iff 
$\rho,i \not\modelsLTL \psi$;\\
$\rho,i \modelsLTL \X\psi$ & iff $i<n$ and 
 $\rho,i{+}1 \modelsLTL \psi$;\\
$\rho,i \modelsLTL \G\psi$ & iff 
$\rho,i \modelsLTL \psi$ and ($i=n$ or
$\rho,i{+}1\modelsLTL \G\psi$);\\
$\rho,i \modelsLTL \psi \U \psi'$ & iff 
$\rho,i \modelsLTL \psi'$ or ($i\,{<}\,n$ and
$\rho,i \modelsLTL \psi$ and
$\rho,i{+}1\modelsLTL \psi \U \psi'$).
\end{tabular}
\end{definition}

Our approach to $\LTLconf$ verification proceeds along the lines of the LTL$_f$
procedure from~\cite{ada}, with the difference that simple constraint atoms are
 replaced by configuration maps.
In order to express the requirements on a run of a DDSA $\BB$ to satisfy an $\LTLconf$ formula $\chi$,
we use a nondeterministic automaton (NFA) $\NFApsi = (Q, \Sigma, \varrho, q_0, Q_F)$, where
the states $Q$ are a set of subformulas of $\psi$,
$\Sigma\,{=}\, 2^{\Confs}$ is the alphabet,
$\varrho$ is the transition relation,
$q_0 \in Q$ is the initial state, and
$Q_F\subseteq Q$ is the set of final states.
The construction of $\NFAfor[\psi]$ is standard~\cite{GMM14,ada}, 
treating configuration maps for the time being as propositions;
but for completeness it is described in \appref{nfa}.
For instance, for a configuration map $K$, $\psi = \F K$ corresponds to the NFA
\smash{\raisebox{-1mm}{\begin{tikzpicture}[node distance=20mm]
 \node[state, minimum width=6mm] (0) {$\psi$};
 \node[state, right of=0, double, minimum width=6mm] (1) {$\top$};
 \draw[edge] (0) -- node[action, above] {$K$} (1);
\draw[edge] ($(0) + (-.4,-.2)$) -- (0);
\draw[->] (0) to[loop left, in=170, out=190, looseness=8] (0);
\draw[->] (1) to[loop left, in=-20, out=20, looseness=8] (1);
\end{tikzpicture}}}
and $\psi' = \X K$ to
\smash{\raisebox{-1mm}{\begin{tikzpicture}[node distance=15mm]
 \node[state, minimum width=6mm] (0) {$\psi'$};
 \node[state, right of=0, minimum width=6mm] (1) {$K$};
 \node[state, right of=1, double, minimum width=6mm] (2) {$\top$};
\draw[edge] ($(0) + (-.4,0)$) -- (0);
 \draw[edge] (0) -- (1);
 \draw[edge] (1) -- node[action, above] {$K$} (2);
\draw[->] (2) to[loop right, looseness=6]  (2);
\end{tikzpicture}}}.
(For simplicity, edges labels $\{K\}$ are shown as $K$,
and edge labels $\emptyset$ are omitted.)

For $w_i\in \Sigma$, i.e., $w_i$ is a set of configuration maps, 
$w_i(b)$ denotes the formula $\bigwedge_{K\in w}K(b)$.
Moreover, for $w = \seqz[n]{w}\in \Sigma^*$ and a symbolic run $\sigma\colon b_0 \goto{a_1} b_1 \goto{a_2} \dots \goto{a_n} b_n$, let
$w\otimes \sigma$ denote the sequence of formulas $\langle w_0(b_0), \dots, w_{n}(b_n)\rangle$, 
i.e., the component-wise application of $w$ to the control states of $\sigma$.
A word $\seqz{w}\in \Sigma^*$ is \emph{consistent} with a run 
$(b_0, \alpha_0) 
\goto{a_1} (b_1, \alpha_1)
\goto{a_2} \dots
\goto{a_n} (b_n, \alpha_n)$
if $\alpha_i \models w_i(b_i)$ for all $i$, $0\,{\leq}\,i\,{\leq}\,n$. 
The key correctness property of $\NFAfor[\psi]$ is the following (cf.~\cite[Lem. 4.4]{ada}, and see \appref{nfa} for the proof adapted to $\LTLconf$):

\begin{lemma}
\label{lem:nfa}
$\NFAfor[\psi]$ accepts a word that is consistent with a run $\rho$ iff 
$\rho \modelsLTL \psi$.
\end{lemma}

\paragraph{Product Construction.}
As a next step in our verification procedure, given a control state $b$ of $\BB$, we aim to find (a symbolic representation of) all configurations $(b,\alpha)$ that satisfy an $\LTLconf$ formula $\psi$.
To that end, we combine $\NFAfor[\psi]$ with $\BB$ to a cross-product automaton $\smash{\NN^\psi_{\BB,b}}$.
For technical reasons, when performing the product construction, the steps in $\BB$ need to be shifted by one with respect to the steps in $\NFAfor[\psi]$. 
Hence, given $b\inn B$, let $\BB_b$ be the \mydds obtained from $\BB$ by adding
a dummy initial state $\smash{\bdummy}$, so that $\BB_b$ has state set $\smash{B' = B\cup\{\bdummy\}}$ and 
transition relation $\smash{T' = T \cup \{(\bdummy, a_0, b)\}}$ for a fresh action $a_0$ with $\guard(a_0) = \top$.

\begin{definition}
\label{def:product construction}
The \emph{product automaton} $\smash{\NN^\psi_{\BB,b}}$ is defined
for an $\LTLconf$ formula $\psi$, a DDSA $\BB$, and a control state $b\in B$.
Let  $\BB_b = \langle B', \bdummy, \AA, T', B_F, V, \alphainit, \guard\rangle$ and $\NFAfor[\psi]$ as above. Then $\smash{\NN^\psi_{\BB,b}=(P, R, p_0, P_F)}$ is as follows:
\begin{compactitem}[$\bullet$]
\item
$P \subseteq B' \times Q \times \CC(V\cup V_0)$, i.e., 
states in $P$ are triples $(b, q, \varphi)$ such that 
\item
the initial state is $p_0=(\bdummy, q_0, \phi_\nu)$; 
\item
if $b \goto{a} b'$ in $T'$, $q \goto{w} q'$ in $\NFAfor[\psi]$, 
and $\update(\varphi, a) \wedge w(b')$ is satisfiable,
there is a transition  $(b,q, \varphi) \goto{a,w} (b',q', \varphi')$ in $R$ 
such that $\varphi' \equiv \update(\varphi, a) \wedge w(b')$;
\item
$(b',q', \varphi')$ is in the set of final states $P_F \subseteq P$ iff $b'\in B_F$, and $q'\in Q_F$.
\end{compactitem}
\end{definition}

\begin{example}
\label{exa:pc}
Consider the DDSA $\BB$ from \exaref{ddsas}, and let
$K = \{ \m b_1 \mapsto \bot,\ \m b_2 \mapsto x\,{\geq}\,2 \wedge y\,{\geq}\,2,\ \m b_3 \mapsto x\,{\geq}\,2\}$. The property $\psi = \X K$ is captured by the NFA 
\smash{\raisebox{-1mm}{\begin{tikzpicture}[node distance=15mm]
 \node[state, minimum width=6mm] (0) {$\psi$};
 \node[state, right of=0, minimum width=6mm] (1) {$K$};
 \node[state, right of=1, double, minimum width=6mm] (2) {$\top$};
\draw[edge] ($(0) + (-.4,0)$) -- (0);
 \draw[edge] (0) -- (1);
 \draw[edge] (1) -- node[action, above] {$K$} (2);
\draw[->] (2) to[loop right, looseness=6]  (2);
\end{tikzpicture}}}.
The product automata $\smash{\NN^\psi_{\BB,\m b_1}}$ and $\smash{\NN^\psi_{\BB,\m b_2}}$ are as follows:

\smallskip
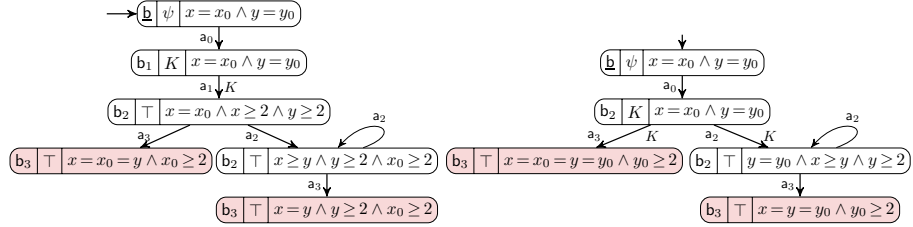
\resizebox{\columnwidth}{!}{
\begin{tikzpicture}[node distance=10mm,>=stealth']
\tikzstyle{node} = [draw,rectangle split, rectangle split parts=3,rectangle split horizontal, rectangle split draw splits=true, inner sep=3pt, scale=.7, rounded corners]
\tikzstyle{goto} = [->]
\tikzstyle{action}=[scale=.6, black]
\tikzstyle{constr}=[scale=.5, black]
\tikzstyle{accepting state} = [fill=red!80!black!15]
\node[node] (0)  {\pcnode{$\underline{\m b}$}{$\psi$}{$x\eqn x_0 \wedge y \eqn y_0$}};
\node[node, below of = 0] (1)
  {\pcnode{$\m b_1$}{$K$}{$x\eqn x_0 \wedge y \eqn y_0 $}};
\node[node, below of = 1] (2)
  {\pcnode{$\m b_2$}{$\top$}{$x\eqn x_0 \wedge x\,{\geq}\,2\wedge y\,{\geq}\,2$}};
\node[node, below of = 2, xshift=-22mm, accepting state] (3a)
  {\pcnode{$\m b_3$}{$\top$}{$x\eqn x_0 \eqn y \wedge x_0\,{\geq}\,2$}};
\node[node, below of = 2, xshift=22mm] (3b)  
  {\pcnode{$\m b_2$}{$\top$}{$x\,{\geq}\,y \wedge y\,{\geq}\,2 \wedge x_0\,{\geq}\,2$}};
\node[node, below of = 3b, accepting state] (4)
  {\pcnode{$\m b_3$}{$\top$}{$x\eqn y \wedge y\,{\geq}\,2\wedge x_0\,{\geq}\,2$}};
\draw[edge] ($(0) + (-1.6,0)$) -- (0);
\draw[edge] (0) -- node[action, left] {$\m a_0$} (1);
\draw[edge] (1) -- node[action, right, yshift=0mm] {$K$} node[action, left] {$\m a_1$} (2);
\draw[edge] (2) -- node[action, left, xshift=-2mm] {$\m a_3$} (3a);
\draw[edge] (2) -- node[action, left, xshift=-2mm] {$\m a_2$} (3b);
\draw[edge] (3b) -- node[action, left] {$\m a_3$} (4);
\draw[->] (3b) to[loop, min distance=4mm, out=20, in=50, looseness=5] node[action, above] {$\m a_2$} (3b);
\end{tikzpicture}
\begin{tikzpicture}[node distance=10mm,>=stealth']
\tikzstyle{node} = [draw,rectangle split, rectangle split parts=3,rectangle split horizontal, rectangle split draw splits=true, inner sep=3pt, scale=.7, rounded corners]
\tikzstyle{goto} = [->]
\tikzstyle{action}=[scale=.6, black]
\tikzstyle{constr}=[scale=.5, black]
\tikzstyle{accepting state} = [fill=red!80!black!15]
\node[node] (0)  {\pcnode{$\underline {\m b}$}{$\psi$}{$x\eqn x_0 \wedge y \eqn y_0$}};
\node[node, below of = 0] (1)
  {\pcnode{$\m b_2$}{$K$}{$x\eqn x_0 \wedge y \eqn y_0 $}};
\node[node, below of = 1, xshift=-24mm, accepting state] (2a)
  {\pcnode{$\m b_3$}{$\top$}{$x\eqn x_0 \eqn y \eqn y_0 \wedge y_0\,{\geq}\,2$}};
\node[node, below of = 1, xshift=24mm] (2b)  
  {\pcnode{$\m b_2$}{$\top$}{$y \eqn y_0\wedge x\,{\geq}\,y \wedge y\,{\geq}\,2$}};
\node[node, below of = 2b, accepting state] (3)
  {\pcnode{$\m b_3$}{$\top$}{$x \eqn y \eqn y_0 \wedge y_0\,{\geq}\,2$}};
\draw[edge] ($(0) + (0,.4)$) -- (0);
\draw[edge] (0) -- node[action, left] {$\m a_0$} (1);
\draw[edge] (1) -- node[action, right, xshift=4mm] {$K$} node[action, left, xshift=-4mm] {$\m a_3$} (2a);
\draw[edge] (1) -- node[action, right, xshift=4mm] {$K$} node[action, left, xshift=-4mm] {$\m a_2$} (2b);
\draw[edge] (2b) -- node[action, left] {$\m a_3$} (3);
\draw[->] (2b) to[loop, min distance=4mm, out=20, in=50, looseness=5] node[action, above] {$\m a_2$} (2b);
\end{tikzpicture}
}
where the shaded nodes are final. 
The formulas in nodes were obtained by applying quantifier
elimination to the formulas built using $\update$ according to
\defref{product construction}.
$\smash{\NN^\psi_{\BB,\m b_3}}$ consists only of the dummy transition and has no final states. 
\end{example}

\defref{product construction} need not terminate if infinitely many non-equivalent formulas occur in the construction. In \secref{modelchecking} we will identify a criterion that guarantees termination. Beforehand, we state the key correctness property, which lifts~\cite[Thm. 4.7]{ada} to LTL with configuration maps. Its proof is similar to the respective result in \cite{ada}, but we provide it in the appendix for completeness.

\begin{theorem}
\label{thm:model:checking}
Let $\psi \inn \LTLconf$ and $b\inn B$ such that there is a finite product automaton
$\smash{\NN_{\BB,b}^\psi}$.
Then there is a final run $\rho \colon (b,\alpha_0) \to^* (b_F, \alpha_F)$ of $\BB$ 
such that $\rho \modelsLTL \psi$, iff
$\smash{\NN_{\BB,b}^\psi}$ has a final state $(b_F, q_F, \phi)$ for some $q_F$ and $\phi$
such that $\phi$ is satisfied by assignment $\gamma$ with
$\gamma(\vec {V_0}) \eqn \alpha_0(\vec V)$ and $\gamma(\vec V) \eqn \alpha_F(\vec V)$.
\end{theorem}

Thus, witnesses for $\psi$ correspond to paths to final states in the product automaton:
e.g., in $\smash{\NN^\psi_{\BB,\m b_1}}$ in \exaref{pc} the formula in the left final node is satisfied by $\gamma(x_0) = \gamma(x) = \gamma(y) = 3$
and $\gamma(y_0) = 0$.
For $\alpha_0$ and $\alpha_2$ such that $\alpha_0(\vec V) = \gamma(\vec V_0) = \{x\mapsto 3, y\mapsto 0\}$ and
$\alpha_2(\vec V) = \gamma(\vec V)= \{x\mapsto 3, y\mapsto 3\}$ there is a witness run for $\psi$ from 
$(\m b_1, \alpha_0)$ to $(\m b_1, \alpha_2)$, e.g.,
$\smash{(\m b_1, \domino{x}{3}{y}{0}) \goto{\m{a}_1}
(\m b_{2}, \domino{x}{3}{y}{3})\goto{\m{a}_3}
(\m b_{3}, \domino{x}{3}{y}{3})}$.

\section{Model Checking Procedure}
\label{sec:modelchecking}

We use the results of the previous section to define a model checking procedure for \CTL formulas, shown in \figref{algs}.
First, we explain the tasks achieved by the three
mutually recursive functions.

$\bullet~$\textit{$\chS(\chi)$} returns a configuration map representing the set of configurations that satisfy a state formula $\chi$.
In the base cases, it returns a function that checks the respective condition, for boolean operators we recurse on the arguments, and 
for a formula $\E \psi$ we proceed to the $\chP$ procedure.

$\bullet~$\textit{$\chP(\psi)$} returns a configuration map $K$ that
represents all configurations which admit a path that satisfies 
$\psi$.
First, $\toLTLconf$ is used to obtain an equivalent $\LTLconf$ formula $\psi'$ (which entails the computation of solutions for all subproperties $\E\eta$).
Then solution $K$ is constructed as follows:
For every control state $b$, we build the product automaton $\smash{\NN^{\psi'}_{\BB,b}}$, and collect the set $\Phi_F$ of formulas in final states. Every $\phi \in \Phi_F$ 
encodes runs from $b$ to a final state of $\BB$ that satisfy $\psi'$.
The variables $\vec {V_0}$ and $\vec V$ in $\phi$ act as placeholders
for the initial and the final values of the runs, respectively.
By $\phi(\vec V, \vec U)$ we rename variables to
use instead $\vec V$ at the start and $\vec U$ at the end,
we quantify existentially over $\vec U$ (as the final valuation is irrelevant),
and take the disjunction over all $\phi \in \Phi_F$. 
The resulting formula $\phi'$ encodes all final runs from $b$ that satisfy $\psi'$, so we set $K(b):=\phi'$.

$\bullet~$\textit{$\toLTLconf(\psi)$} computes an $\LTLconf$ formula equivalent to a path formula $\psi$. To this end, it performs two kinds of replacements in $\psi$:
(a) $\top$, $b\inn B$, and constraints $c$ are represented as configuration maps; and
(b) subformulas $\E\eta$ are replaced by their solutions $\Conf_{\m E\eta}$, which are computed by a recursive call to $\chP$.

\begin{figure}[ht]
\small
\begin{algorithmic}[1]
\Procedure{$\chS$}{$\chi$}
\Switch{$\chi$}
\State \textbf{case} $\top$, $b\in B$, or $c\in \mathcal C$:
  \textbf{return} $K_\chi$ 
\State \casebox{$\chi_1 \wedge \chi_2$}
  \textbf{return} $\chS(\chi_1) \wedge \chS(\chi_2)$ 
\State \casebox{$ \neg \chi$}
  \textbf{return} $\neg \chS(\chi)$ 
\State \casebox{$ \E\psi$}
  \textbf{return} $\chP(\psi)$
\EndSwitch
\EndProcedure
\end{algorithmic}
\begin{algorithmic}[1]
\Procedure{$\chP$}{$\psi$}
\State{$\psi':=\toLTLconf(\psi)$}
\For{$b\in B$}
\State{$(P, R, p_0, P_F) := \NN^{\psi'}_{\BB,b}$}
\Comment{product automaton for $\psi'$, $\BB$, and $b$}
\State{$\Phi := \{ \phi \mid (b_F,q_F,\phi) \in P_F\}$}
\Comment{collect formulas in final states}
\State{$\Conf(b) := \bigvee_{\phi \in \Phi} \exists \vec U. \phi(\vec V,\vec U)$}
\EndFor
\State \textbf{return} $K$
\EndProcedure
\end{algorithmic}
\begin{algorithmic}[1]
\Procedure{$\toLTLconf$}{$\psi$}
\Switch{$\psi$}
\State \textbf{case} $\top$, $b\in B$, or $c\in \mathcal C$:
  \textbf{return} $K_\psi$ 
\State \casebox{$\psi_1 \wedge \psi_2$}
  \textbf{return} $\toLTLconf(\psi_1) \wedge \toLTLconf(\psi_2)$ 
\State \casebox{$ \neg \psi$}
  \textbf{return} $\neg \toLTLconf(\psi)$
\State \casebox{$ \E\psi$}
  \textbf{return} $\chP(\psi)$
\State \casebox{$\X\psi$}
  \textbf{return} $\X\:\toLTLconf(\psi)$
\State \casebox{$\G\psi$}
  \textbf{return} $\G\:\toLTLconf(\psi)$
\State \casebox{$\psi_1 \U \psi_2$}
  \textbf{return} $\toLTLconf(\psi_1) \U \toLTLconf(\psi_2)$
\EndSwitch
\EndProcedure
\end{algorithmic}
\caption{Model checking procedure.\label{fig:algs}}
\end{figure}

To represent the base cases of formulas as configuration maps in \figref{algs}, we define
$K_\top :=(\lambda \_.\:\top)$,
$K_b := (\lambda b'.\:b \eqn b'\:?\:\top : \bot)$ for all $b\inn B$,
and
$K_c := (\lambda \_.\:c)$ for constraints $c$.
We also write 
$\neg K$ for $(\lambda b. \neg K(b))$ and 
$K \wedge K'$ for $(\lambda b. K(b) \wedge K'(b))$.
The next example illustrates the approach.

\begin{example}
\label{exa:mc}
Consider $\chi \eqn \E \X (\A \G (x\,{\geq}\,2))$ and the DDSA $\BB$ in \exaref{ddsas}.
To get a solution $K_1$ to 
$\chS(\chi) = \chP(\psi_1)$ for $\psi_1 = \X (\A \G (x\,{\geq}\,2))$,
we first compute an equivalent $\LTLconf$ formula $\psi_1' = \X \Conf_2$,
where $\Conf_2$ is a solution to $\A \G (x\,{\geq}\,2) \equiv \neg \E \F (x\,{<}\,2)$.
To this end, we run $\chP(\psi_2)$ for $\psi_2 = \F (x\,{<}\,2)$,
which is represented in $\LTLconf$ as $\psi_2' = \F (\Conf_{x<2})$
with NFA 
\smash{\raisebox{-1mm}{\begin{tikzpicture}[node distance=20mm]
 \node[state, minimum width=6mm] (0) {$\psi_1'$};
 \node[state, right of=0, double, minimum width=6mm] (1) {$\top$};
 \draw[edge] (0) -- node[action, above] {$K_{x < 2}$} (1);
\draw[edge] ($(0) + (-.4,-.2)$) -- (0);
\draw[->] (0) to[loop left, in=170, out=190, looseness=8] (0);
\draw[->] (1) to[loop left, in=-20, out=20, looseness=8] (1);
\end{tikzpicture}}}.
Next, $\chP$ builds $\smash{\NN^{\psi_2'}_{\BB,b}}$ for all  states $b$.
For instance, for $\m b_2$ we get:
\\
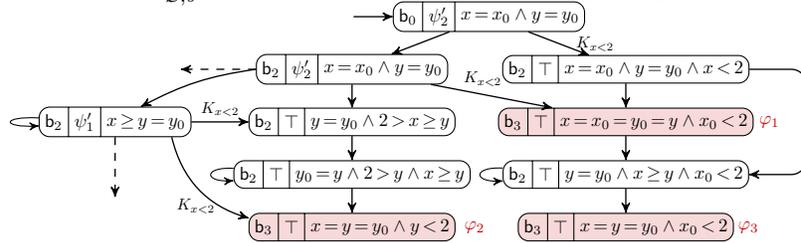
\centerline{
\begin{tikzpicture}[node distance=10mm,>=stealth',rounded corners=2mm]
\tikzstyle{node} = [draw,rectangle split, rectangle split parts=3,rectangle split horizontal, rectangle split draw splits=true, inner sep=3pt, scale=.7, rounded corners]
\tikzstyle{goto} = [->]
\tikzstyle{action}=[scale=.6, black]
\tikzstyle{constr}=[scale=.5, black]
\tikzstyle{accepting state} = [fill=red!80!black!15]
\node[node] (0) {\pcnode{$\m b_0$}{$\psi_2'$}{$x\eqn x_0 \wedge y \eqn y_0$}};
\node[node, below of = 0, xshift=-26mm] (1a)
  {\pcnode{$\m b_2$}{$\psi_2'$}{$x\eqn x_0 \wedge y \eqn y_0 $}};
\node[node, below of = 0, xshift=26mm] (1b)
  {\pcnode{$\m b_2$}{$\top$}{$x\eqn x_0 \wedge y \eqn y_0 \wedge x\,{<}\,2$}};
\node[node, below of = 1a] (2a)
  {\pcnode{$\m b_2$}{$\top$}{$y \eqn y_0 \wedge 2\,{>}\,x\,{\geq}\,y$}};
\node[node, below of = 1b, accepting state] (2b)  
  {\pcnode{$\m b_3$}{$\top$}{$x\eqn x_0 \eqn y_0 \eqn y \wedge x_0\,{<}\,2$}};
\node[node, below of = 2a] (3a)  
  {\pcnode{$\m b_2$}{$\top$}{$y_0 \eqn y \wedge 2\,{>}\,y \wedge x\,{\geq}\,y$}};
\node[node, below of = 1a, xshift=-45mm] (2c)  
  {\pcnode{$\m b_2$}{$\psi_1'$}{$x\,{\geq}\,y \eqn y_0$}};
\node[node, below of = 2b] (3b)
  {\pcnode{$\m b_2$}{$\top$}{$y \eqn y_0 \wedge x\,{\geq}\,y \wedge x_0\,{<}\,2$}};
\node[node, below of = 3a, accepting state] (4a)
  {\pcnode{$\m b_3$}{$\top$}{$x\eqn y \eqn y_0 \wedge y\,{<}\,2$}};
\node[node, below of = 3b, accepting state] (4b)
  {\pcnode{$\m b_3$}{$\top$}{$x\eqn y \eqn y_0 \wedge x_0\,{<}\,2$}};
\draw[edge] ($(0) + (-1.8,0)$) -- (0);
\draw[edge] (0) -- (1a);
\draw[edge] (0) -- node[action, above, very near end, xshift=3mm, yshift=-1mm] {$\Conf_{x<2}$} (1b);
\draw[edge] (1a) -- (2a);
\draw[edge] (1b) -- (2b);
\draw[edge,dashed] (1a) -- ($(1a.west) + (-1,0)$);
\draw[edge,dashed] (2c) -- ($(2c) + (0,-1)$);
\draw[edge] (1a) -- node[action, above, xshift=-2mm, yshift=1mm] {$\Conf_{x<2}$} (2b);
\draw[edge] (2a) --  (3a);
\draw[edge] (2b) -- (3b);
\draw[edge, bend right=10] (1a.west) to (2c);
\draw[edge] (1b.east) -| ($(1b) + (2.4,-1)$) |- (3b.east);
\draw[->] (2c) to[loop left, min distance=5mm, out=177, in=183, looseness=8] (2c);
\draw[edge] (2c) --node[action, above, yshift=0mm] {$\Conf_{x<2}$} (2a);
\draw[edge, bend right] (2c.-15) to node[action, left, near end, yshift=0mm] {$\Conf_{x<2}$} (4a.175);
\draw[edge] (3a) -- (4a);
\draw[edge] (3b) -- (4b);
\draw[->] (3b) to[loop left, min distance=4mm, out=177, in=183, looseness=4] (3b);
\draw[->] (3a) to[loop, min distance=4mm, out=177, in=183, looseness=4] (3a);
\node[right of = 2b, scale=.7, red!80!black, xshift=13mm] {$\phi_1$};
\node[right of = 4a, scale=.7, red!80!black, xshift=9mm] {$\phi_2$};
\node[right of = 4b, scale=.7, red!80!black, xshift=9mm] {$\phi_3$};
\end{tikzpicture}}\\
\noindent
where dashed arrows indicate transitions to
non-final sink states.
For $\vec U = \langle \hat x, \hat y\rangle$, and the formulas $\phi_1$, $\phi_2$, and $\phi_3$
in final nodes, we compute\\[.5ex]
\centerline{
$\begin{array}{rll}
\exists \vec U.\: \phi_1(\vec V, \vec U) &= \exists \hat x\,\hat y.\ 
 \hat x \eqn x \eqn \hat y \eqn y\wedge x < 2 &\equiv x < 2 \\
\exists \vec U.\: \phi_2(\vec V, \vec U) &= \exists \hat x\,\hat y.\ 
  \hat x \eqn \hat y \eqn y \wedge \hat y\,{<}\, 2 &\equiv y < 2 \\
\exists \vec U.\: \phi_3(\vec V, \vec U) &= \exists \hat x\,\hat y.\ 
  \hat x \eqn \hat y \eqn y \wedge x\,{<}\,2 &\equiv x < 2
\end{array}$}\\
so that $K_3:=\chP(\psi_2)$ sets
$K_3(\m b_2) = \bigvee_{i=1}^3 \exists \vec U.\: \phi_i(\vec V, \vec U) \equiv 
x\,{<}\,2 \vee y\,{<}\,2$.
For reasons of space, the constructions for $\m b_1$ and $\m b_3$
are shown in \exaref{mc:continued} in \appref{examples};
we obtain $K_3(\m b_1) = \top$ and $K_3(\m b_3) = x < 2$.
By negation, the solution $K_2$ to $\A \G (x \geq 2)$ is
$K_2 = \neg K_3 = \{ \m b_1 \mapsto \bot,\ \m b_2 \mapsto x \geq 2 \wedge y \geq 2,\ \m b_3 \mapsto x \geq 2\}$.
Now we can proceed with $\chP(\psi_1)$.
The NFA and product automata for $\psi_1' = \X \Conf_2$ are as shown in \exaref{pc} and
in a similar way as above we obtain
the solution $K_1$ for $\E \X \A \G (x \geq 2)$ as
$K_1 = \{ \m b_1 \mapsto x\,{\geq}\,2,\ \m b_2 \mapsto y\,{\geq}\,2 ,\ \m b_3 \mapsto \bot\}$.
Thus, $\BB$ satisfies the property for any initial assignment $\alphainit$  with $\alphainit(x) \geq 2$.
\end{example}

\noindent
Next we prove correctness of $\chS(\chi)$ under the condition that
it is defined, i.e., all required product automata are finite.
First we state our main result, 
but before giving its proof we show helpful properties of $\toLTLconf$ and $\chP$.

\begin{theorem}
\label{thm:main}
 For every configuration $(b,\alpha)$ of the DDSA $\BB$ and every  state property $\chi$,
 if $\chS(\chi)$ is defined then
$(b,\alpha) \models \chi$ iff $\alpha \models \chS(\chi)(b)$.
\end{theorem}

\begin{lemma}
\label{lem:toLTL}
Let $\psi$ be a path formula with $\qd(\psi) \eqn k$.
Suppose that for all confi\-gurations $(b,\alpha)$ and path formulas $\psi'$ with $\qd(\psi')\,{<}\,k$, there is a $\rho' \in \FRuns(b,\alpha)$ with $\rho' \models \psi'$ iff $\alpha \models \chP(\psi')(b)$. Then
$\rho \models \psi$ iff $\rho \modelsLTL \toLTLconf(\psi)$.
\end{lemma}
\begin{proof}[sketch]
By induction on $\psi$.
The base cases are by the definitions of $K_\top$, $K_b$, and $K_c$.
In the induction step, if $\psi \eqn \E \psi'$ then $\rho \models \psi$ iff
$\exists\rho' \inn \FRuns(b_0, \alpha_0)$ with $\rho' \models \psi'$, for $\rho_0 \eqn (b_0, \alpha_0)$. As 
$\qd(\psi')\,{<}\,\qd(\psi)$, this holds by assumption iff
$\alpha_0 \models \chP(\psi')(b_0)$. This is equivalent to $\rho \modelsLTL \toLTLconf(\psi) = \chP(\psi')$.
All other cases are by the induction hypothesis and \defsref{semantics}{ltl:semantics}.
\end{proof}

\begin{lemma}
\label{lem:chP}
If $\psi' = \toLTLconf(\psi)$ such that for all runs $\rho$ it is
$\rho \models \psi$ iff $\rho \modelsLTL \psi'$,
there is a run $\rho \inn \FRuns(b,\alpha)$ with $\rho \models \psi$
 iff $\alpha \models \chP(\psi)(b)$.
\end{lemma}
\begin{proof}
($\Longrightarrow$)
Suppose there is a run $\rho \inn \FRuns(b,\alpha)$ with $\rho \models \psi$,
so $\rho$ is of the form $(b,\alpha)\to^* (b_F,\alpha_F)$ for some $b_F \in B_F$.
By assumption, this implies $\rho \modelsLTL \psi'$, so that
by \thmref{model:checking},
$\smash{\NN_{\BB,b}^{\psi'}}$ has a final state $(b_F, q_F, \phi)$
where $\phi$ is satisfied by an assignment $\gamma$ with domain $V\cup V_0$ such that
$\gamma(\vec {V_0}) \eqn \alpha(\vec V)$ and $\gamma(\vec V) \eqn \alpha_F(\vec V)$.
By definition, $\chP(\psi)(b)$ contains a disjunct $\exists \vec U.\:\phi(\vec V, \vec U)$. As $\gamma$ satisfies $\phi$ and $\gamma(\vec {V_0}) \eqn \alpha(\vec V)$,
$\alpha \models \chP(\psi)(b)$.
($\Longleftarrow$) If $\alpha \models \chP(\psi)(b)$,
by definition of $\chP$ there is a formula $\phi$
such that $\alpha \models \exists \vec U.\:\phi(\vec V, \vec U)$
and $\phi$ occurs in a final state $(b_F, q_F, \phi)$ of $\smash{\NN_{\BB,b}^{\psi'}}$. 
Hence there is an assignment $\gamma$ with domain $V\cup V_0$ and 
$\gamma(\vec {V_0}) \eqn \alpha(\vec V)$ such that $\gamma \models \phi$.
By \thmref{model:checking}, there is a run $\rho\colon(b, \alpha) \to^* (b_F, \alpha_F)$ such that $\rho \modelsLTL \psi'$.
By the assumption, we have $\rho \models \psi$.
\qed
\end{proof}

\noindent
At this point the main theorem can be proven:

\begin{proof}[of \thmref{main}]
We first show ($\star$): for any path formula $\psi$,
there is a run $\rho \in \FRuns(b,\alpha)$ with $\rho \models \psi$ iff $\alpha \models \chP(\psi)(b)$.
The proof is by induction on $\qd(\psi)$.
If $\psi$ contains no path quantifiers, \lemref{toLTL} implies that
$\rho \models \psi$ iff $\rho \modelsLTL \toLTLconf(\psi)$ for all runs $\rho$, so ($\star$)
follows from \lemref{chP}.
In the induction step, we conclude from \lemref{toLTL}, using the induction hypothesis of ($\star$) as assumption, that $\rho \models \psi$ iff $\rho \modelsLTL \toLTLconf(\psi)$ for all runs $\rho$. Again, ($\star$) follows from \lemref{chP}.

The theorem is then shown by induction on $\chi$:
The base cases $\top$, $b'\inn B$, $c\inn \CC$ are easy to check, and for 
properties of the form $\neg \chi'$ and $\chi_1 \wedge \chi_2$ the claim
follows from the induction hypothesis and the definitions.
Finally, for $\chi = \E \psi$, $(b,\alpha) \models \chi$ iff there is a run $\rho \in \FRuns(b,\alpha)$ such that $\rho \models \psi$. By ($\star$) this is the case iff 
$\alpha \models \chP(\psi)(b) = \chS(\chi)(b)$.
\qed
\end{proof}

\paragraph{Termination}
We next show that the formulas generated in our procedure all
have a particular shape, to obtain an abstract termination result.
For a set of formulas $\Phi \subseteq \formulas$  and a symbolic run 
$\sigma$, let a history constraint $\hist(\sigma,\vseq)$ be \emph{over basis $\Phi$}
if $\vseq= \langle \theta_0,\dots, \theta_n\rangle$ and 
for all $i$, $1\,{\leq}\,i\,{\leq}\,n$, there is a subset $T_i \subseteq \Phi$
s.t. $\theta_i = \bigwedge T_i$.
Moreover, for a set of formulas $\Phi$, let
$\Phi^\pm = \Phi \cup \{ \neg \phi \mid \phi \in \Phi\}$.

\begin{definition}
For a DDSA $\BB$, a constraint set $\CC$ over free variables $V$, and $k\,{\geq}\, 0$, the formula sets $\Phi_k$ are inductively defined by
$\Phi_0 = \CC \cup \{\top, \bot\}$ and
\[\Phi_{k+1} = 
\{ {\bigvee}_{\phi \in H} \;\; \exists \vec U.\ \phi(\vec V, \vec U)\mid H\,{\subseteq}\, \HH_k\}\]
where $\HH_k$ is the set  of all history constraints of $\BB$ with basis $\bigcup_{i\leq k}\Phi_{i}^\pm$.
\end{definition}

\noindent
Note that formulas in $\Phi_{k}$ have free variables $V$, while
those in $\HH_k$ have free variables $V_0 \cup V$.
We next show that these sets correspond to the formulas generated by our procedure,
if all constraints in the verification property are in $\CC$.

\begin{lemma}
\label{lem:formula:sets}
Let $\E\psi$ have quantifier depth $k$, $\psi' = \toLTLconf(\psi)$, and 
$\smash{\NN_{\BB,b}^{\psi'}}$ be a constraint graph constructed in $\chP(\psi)$ for some $b\in B$.
Then, 
\begin{compactenum}[(1)]
\item for all nodes $(b', q, \phi)$ in $\smash{\NN_{\BB,b}^{\psi'}}$ there is some $\phi' \in \HH_k$ such that $\phi \equiv \phi'$, 
\item $\chP(\psi)(b)$ is equivalent to a formula in $\Phi_{k+1}$.
\end{compactenum}
\end{lemma}

\noindent
The statements are proven by induction on $k$, using the results about the product construction (\lemref{pc}).
From part (1) of this lemma and \thmref{main} we thus obtain an abstract criterion for decidability that will become useful in the next section:

\begin{corollary}
\label{cor:termination}
For a DDS $\BB$ as above and a state formula $\chi$, if $\HH_j(b)$ 
is finite up to equivalence for all 
$j\,{<}\,\qd(\chi)$ and $b\inn B$, the verification problem is solvable.
\end{corollary}
\begin{proof}
By the assumption about the sets $\HH_j(b)$ for $j\,{<}\,\qd(\chi)$, 
all product automata constructions in
recursive calls $\chP(\psi)$ of $\chS(\chi)$ terminate if logical equivalence of  formulas is checked eagerly.
Thus $\chS(\chi)$ is defined, and by \thmref{main}, it solves the verification problem.
\qed
\end{proof}

The property that all sets $\HH_j(b)$, $j\,{<}\,\qd(\chi)$, are finite might not be decidable itself. However, in the next section we will show means to guarantee this property. Moreover, we remark that finiteness of all $\HH_j(b)$  implies a \emph{finite history set}, a decidability criterion identified for the linear-time case \cite[Def. 3.6]{ada}; but \exaref{bounded:lookback} below illustrates that the requirement on the $\HH_j(b)$'s is strictly stronger.

\section{Decidability of \mydds Classes}
\label{sec:criteria}

We here illustrate restrictions on DDSAs, either on the control flow or on the constraint language, that render our approach a decision procedure for CTL$^*_f$.

\medskip
\noindent
\textbf{Monotonicity constraints} (MCs) restrict constraints (\defref{constraint}) as follows: MCs over variables $V$ and domain $D$ have the form $p \odot q$ where $p,q\in {D\,{\cup}\,V}$
and $\odot$ is one of $=, \neq, \leq, <, \geq$, or $>$.
The domain $D$ may be $\mathbb R$ or $\mathbb Q$.
We call a boolean formula whose atoms are MCs an \emph{MC formula},
a DDSA where all atoms in guards are MCs an \emph{MC-DDSA},
and a \CTLsf property whose constraint atoms are MCs an \emph{MC property}.
For instance, $\BB$ in \exaref{ddsas} is an MC-DDSA.

We exploit a useful quantifier elimination property:
If $\phi$ is an MC formula over a set of constants $\Const$ and variables $V \cup \{x\}$, there is some $\phi' \equiv \exists x.\, \phi$ such that $\phi'$ is a quantifier-free MC formula over $V$ and $\Const$. Such a $\phi'$ can be obtained by writing $\phi$ in disjunctive normal form and applying a Fourier-Motzkin procedure~\cite[Sec. 5.4]{KS16} to each disjunct, which guarantees that all constants in $\phi'$ also occur in $\phi$.

\begin{theorem}\label{thm:mc}
The verification problem is solvable for all combinations of an MC-DDSA $\BB$ and an MC property $\chi$.
\end{theorem}
\begin{proof}
\newcommand{\MCC}{\MC_{\!\Const}}
Let $\chi$ be an MC property, and $\Const$ the finite set of constants in constraints in $\chi$, ${\alpha_0}$, and guards of $\BB$.
Let moreover $\MCC$ be the set of quantifier-free formulas whose atoms are MCs over $V \cup V_0$ and $\Const$, so $\MCC$ is finite up to equivalence.

We show the following property ($\star$): all history constraints $\hist(\sigma,\vec \theta)$ 
over basis $\MCC$ are equivalent to a formula in $\MCC$.
For a symbolic run 
$\sigma\colon b_0 \to^* b_{n-1}\goto{a} b_n$ and a sequence
$\vseq= \langle \theta_0,\dots, \theta_n\rangle$ over $\MCC$, the proof is by induction on $n$.
In the base case, $\hist(\sigma, \vseq) \eqn \phi_\nu \wedge \theta_0$ is in $\MCC$ because
$\phi_\nu$ is a conjunction of equalities between $V\cup V_0$, and $\theta_0 \in \MCC$ by assumption.
In the induction step, $\hist(\sigma, \vseq) \eqn \update(\hist(\sigma|_{n-1}, \vseq|_{n-1}), a_{n}) \wedge \theta_n$. By induction hypothesis, $\hist(\sigma|_{n-1}, \vseq|_{n-1})\equiv \phi$ for some $\phi$ in $\MCC$.
Thus $\hist(\sigma, \vseq) \equiv \exists \vec U. \phi(\vec U) \wedge \Delta_a(\vec U, \vec V) \wedge \theta_n$. As $\BB$ is an $\MC$-DDSA, $\Delta_a(\vec U, \vec V)$ is a
conjunction of MCs over $V\cup U$ and constants $\Const$, and $\theta_n \in \MCC$ by assumption.
By the quantifier elimination property, there exists a quantifier-free MC-formula $\phi'$ over variables $V_0\cup V$ that is equivalent to $\exists \vec U. \phi(\vec U) \wedge \Delta_a(\vec U, \vec V) \wedge \theta_n$, and mentions only constants in $\Const$, so $\phi' \in \MCC$.

For $\CC$ the set of constraints in $\chi,$
we now show that $\HH_j \subseteq \MCC$ for all $j\geq 0$, by induction on $j$.
In the base case ($j\eqn 0$),
the claim follows from ($\star$), as all constraints in $\Phi_0$, i.e., in $\chi$, are in $\MCC$.
For $j\,{>}\,0$, consider first a formula $\widehat \phi \in\Phi_{j}$ for some $b\inn B$.
Then $\widehat \phi$ is of the form
$\widehat \phi = \bigvee_{\phi \in H} \exists \vec U.\ \phi(\vec V, \vec U)$ for some
$H\,{\subseteq}\, \HH_{j-1}$. By the induction hypothesis, $H\subseteq \MCC$,
so by the quantifier elimination property of MC formulas, $\widehat \phi$ is equivalent
to an MC-formula over $V$ and $\Const$ in $\MCC$. As $\HH_j$ s built over basis $\Phi_{j}$, the
claim follows from $(\star)$.
\qed
\end{proof}

Notably, the above quantifier elimination property fails for MCs
over integer variables; indeed, CTL model checking is undecidable
 in this case~\cite[Thm. 4.1]{MT16}.

\medskip
\noindent
\textbf{Integer periodicity constraint} systems confine the constraint language to variable-to-constant comparisons and restricted forms of variable-to-variable comparisons, and are for instance used in calendar formalisms~\cite{Demri06,DemriG08}.
More precisely, \emph{integer periodicity constraint} (IPC) atoms have the form $x = y$, $x \odot d$ for $\odot \in \{=,\neq, <, >\}$, $x \equiv_k y + d$, or $x \equiv_k d$, for variables $x,y$ with domain $\mathbb Z$ and $k,d\in \mathbb N$. 
A boolean formula whose atoms are IPCs is an \emph{IPC formula},
a DDSA whose guards are conjunctions of IPCs an \emph{IPC-DDSA},
and a \CTLsf formula whose constraint atoms are IPCs an \emph{IPC property}.
For instance, $\BB_{\mathit{ipc}}$ in \exaref{ddsas} is an IPC-DDSA.

Using \corref{termination} and a known quantifier elimination property for IPCs \cite[Thm. 2]{Demri06}, one can show that the verification problem is also solvable for
IPC-DDSAs, in a proof that resembles the one of \thmref{mc} (see \appref{proofs}).

\begin{theorem}\label{thm:ipc}
The verification problem is solvable for all combinations of an IPC-DDSA $\BB$ and an IPC-property $\chi$.
\end{theorem}


\noindent
\textbf{Bounded lookback systems}~\cite{ada} restrict the control flow of the DDSA rather than the constraint language, and is a generalization of the earlier criterion of \emph{feedback-freedom}~\cite{DDV12}. 
Intuitively, the property demands that the behavior of a DDSA at any point in time depends only on boundedly many events from the past.
We refer to~\cite[Def. 5.9]{ada} for the formal definition.
Systems that enjoy bounded lookback 
allow for decidable linear-time verification~\cite[Thm. 5.10]{ada}.
However, we next show that this is not the case for branching time.

\begin{example}
\label{exa:bounded:lookback}
We reduce control state reachability of two-counter machines (2CM) to decidability of \CTLsf 
formulas for feedback-free (and hence bounded lookback) systems, inspired by~\cite[Thm. 4.1]{MT16}. 
2CMs have a finite control structure and two counters $x_1$, and $x_2$
that can be incremented, decremented, and tested for 0.
It is undecidable whether a 2CM will ever reach a designated control state $f$~\cite{Minsky67}.
For a 2CM $\mc M$, we build a feedback-free DDSA $\BB\,{=}\,\langle B, \binit, \AA, T, B_F, V, \alphainit, \guard\rangle$ and a \CTLsf property $\chi$ such that $\BB$ satisfies $\chi$ iff $f$ is reachable in $\mc M$.
The set $B$ consists of the control states of $\mc M$, together with an error state $\mathit{e}$ and auxiliary states $b_t$ for transitions $t$ of $\mc M$, such that $B_F = \{f, e\}$.
The set $V$ consists of $x_1$, $x_2$ and auxiliary variables $p_1$, $p_2$, $m_1$, $m_2$. 
Zero-test transitions of $\mc M$ are directly modeled in $\BB$, whereas a step $q \to q'$ that increments $x_i$ by one is modeled as:\\
%
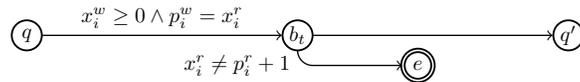
\centerline{
\begin{tikzpicture}[node distance=45mm, rounded corners=2mm]
\tikzstyle{action}=[scale=.8]
\tikzstyle{state}=[draw, circle, inner sep=.5pt, line width=.7pt, scale=.8, minimum width=5mm]
 \node[state] (0) {$q$};
 \node[state, right of = 0] (1) {$b_t$};
 \node[state, right of = 1] (2) {$q'$};
 \node[state, below of = 1, xshift=2cm, yshift=40mm, double] (3) {$e$};
 \draw[edge] (0) -- node[action, above] {$x_i^w \geq 0 \wedge p_i^w = x_i^r$} (1);
 \draw[edge] (1) -- (2);
 \draw[edge] (1) |- node[action, left] {$x_i^r \neq p_i^r + 1$} (3);
\end{tikzpicture}
}\\
The step $q \to b_t$ writes $x_i$, storing its previous value in $p_i$, but if 
the write was not an increment by exactly 1, a step to state $e$ is enabled.
Decrements are modeled similarly.
For $\CC = \emptyset$ and a symbolic run $\sigma$ of $\BB$,
the only possible non-equality edge in $G_{\sigma,\CC}$ is a final step to $e$. Thus, there is no non-equality path between different instants of the same variable, so $\BB$ is feedback-free.
As increments are not exact, $\BB$ overapproximates $\mc M$.
However, $\chi = \E \G(\neg \E \X e)$
asserts existence of a path that never allows for a step to $e$ (i.e., it properly si\-mulates  $\mc M$) but reaches the final state $f$. 
Thus, $\BB$ satisfies $\chi$ iff $f$  
is reachable in $\mc M$.
\end{example}


\section{Implementation}
\label{sec:implementation}

We implemented our approach in the prototype \tool (arithmetic DDS analyzer) in Python;
source code, benchmarks, and a web interface are available (\url{https://ctlstar.adatool.dev}).
The tool takes a CTL$^*$ property $\chi$ together with either a DDSA in JSON format, or a (bounded) Petri net with data (DPN) in PNML format~\cite{pnml} 
as input, in the latter case the system is transformed into a DDSA.
The tool then applies the algorithm in \figref{algs}.
If successful, it outputs the configuration map returned by $\chS(\chi)$,
and it can visualize the product constructions.
To perform SMT checks and quantifier elimination, \tool interfaces CVC5~\cite{DR0BT14} and Z3~\cite{Z3}.
Besides numeric variables, \tool also supports variables of type boolean and string.
In addition to the operations in \defref{ctlstar}, \tool allows next operators $\langle a\rangle$ via an action $a$, which are useful for verification.

We tested \tool on a set of business process models presented as 
Data Petri nets (DPNs) in the literature. As these nets are bounded, they can be transformed into DDSAs. The results are reported in the table below. We indicate
whether the system belongs to a decidable class,
the verified property and whether it is satisfied by the initial configuration,
the verification time, 
the number of SMT checks, and the sizes of both the DDSA $\BB$, and the sum of all product constructions, as numbers of nodes/transitions.
We used CVC5 as SMT solver; times are without visualization, which tends to be time-consuming
for large graphs.
All tests were run on an Intel Core i7 with $4{\times}2.60$GHz and 19GB RAM. 

\begin{center}
\scriptsize
\begin{tabular}{ll@{\,}|c|c|c|r|r|r@{/}l|r@{/}ll}
\multicolumn{2}{c|}{\textbf{process}} & 
\multicolumn{1}{|c}{\textbf{class}} &
\multicolumn{1}{|@{\,}c@{\,}|}{\textbf{property}} & 
\multicolumn{1}{|@{\,}c@{\,}|}{\textbf{sat}} & 
\multicolumn{1}{|c}{\textbf{time}} & 
\multicolumn{1}{|c}{\textbf{checks}} &
\multicolumn{2}{|c}{\textbf{$|\BB|$}} & 
\multicolumn{2}{|c}{\textbf{$\sum |\NN^\psi_{\BB,b}|$}} \\
\hline
(a) & road fines (mined) 
  & MC & no deadlock & no & 7.0s & 8161 &   9&19   &  2052&3067  \\ 
  & &  & $\psi_{a1}$ &  yes & 7.6s & 7655  &  \multicolumn{2}{c|}{ }  & 1987&2906   \\ 
  & &  & $\psi_{a2}$ &  no & 1m12s & 111139  &  \multicolumn{2}{c|}{ }  & 3622&6778   \\ \hline
(b) & road fines (mined) 
  & MC & no deadlock &  yes  & 15m27s & 247563 &   9&19   &  4927&7288  \\ 
  & &  & $\psi_{a1}$ &  yes & 16m7s & 246813  &  \multicolumn{2}{c|}{ } & 4927&7288   \\ \hline 
(c) & road fines (norm) 
  &  & no deadlock & no & 9s & 9179 &   9&19  & 1985&2734  \\
  & &  & $\psi_{a1}$ &  yes & 6.6s & 6382  &  \multicolumn{2}{c|}{ }  & 1597&2167   \\ 
  & &  & $\psi_{c1}$ &  no & 11.5s & 17680  &  \multicolumn{2}{c|}{ }  & 1280&2587   \\ 
  & &  & $\psi_{c2}$ &  no & 10.0s & 15187 &  \multicolumn{2}{c|}{ }  & 1280&2173   \\ 
  & &  & $\psi_{c3}$ &  no & 10.5  & 16000 &  \multicolumn{2}{c|}{ }  & 1280&2240   \\ 
\hline 
(d) & hospital billing 
  & MC,IPC & no deadlock & yes & 20m59s & 1234928 &  17&40  & 23147&38652   \\ 
  & &  & $\psi_{d1}$ & yes & 10m20s & 669379 &  \multicolumn{2}{c|}{ }  & 10654&17415   \\ \hline 
(e) & sepsis (norm) 
  &  & no deadlock & yes & 1m36s & 139 & 301&1630 &  44939&162194 \\  
  & &  & $\psi_{e1}$ & no & 30.1s & 170 &  \multicolumn{2}{c|}{ }  & 22724&81351   \\   
  & &  & $\psi_{e2}$ & yes & 32s & 153 &  \multicolumn{2}{c|}{ }  & 22538&81165   \\ \hline  
(f) & sepsis (mined)
  & MC & no deadlock & yes & 7m24 & 4524 & 301&1630 & 161242&497985 \\ 
  & &  & $\psi_{f1}$ & yes & 3m53s & 5734 &  \multicolumn{2}{c|}{ }  & 74984&237534   \\    \hline  
(g) & board: register 
  &  & no deadlock & yes & 1.4s & 12 & 7&6 & 27&21 \\ \hline  
(h) & board: transfer 
  & MC, IPC & no deadlock & yes & 1.4s & 27      & 7&6 & 51&44 \\ \hline  
(i) & board: discharge 
  &  MC, IPC & no deadlock & yes & 1.5s  & 25 & 6&6 & 67&55 \\ 
  &&& $\psi_{i1}$ & yes & 1.5s & 94 & \multicolumn{2}{c|}{ }  & 91&94 \\
  &&& $\psi_{i2}$ & yes & 1.5s & 27 & \multicolumn{2}{c|}{ }  & 98&102 \\ 
  &&& $\psi_{i3}$ & yes & 1.4s & 56 & \multicolumn{2}{c|}{ }  & 43&43 \\ \hline 
(j) & credit approval 
  &  & no deadlock & yes & 1.7s & 470 & 6&10  & 230&232    \\ 
  &&& $\psi_{j1}$ & yes & 13.2s & 14156 & \multicolumn{2}{c|}{ }  & 645&1324 \\
  &&& $\psi_{j2}$ & no & 3.7s & 3128 & \multicolumn{2}{c|}{ }  & 316&396 \\
  &&& $\psi_{j3}$ & yes & 5.6s & 4748 & \multicolumn{2}{c|}{ }  & 548&655 \\ \hline 
(k) & package handling     
  & MC, IPC & no deadlock & yes & 2.7ss & 1025 & 16&28  & 693&671    \\
  &   
  &  & weak sound ($\tau_1$) & yes & 2.5s & 1079 & \multicolumn{2}{c|}{ }  & 398&382  \\
  &   
  &  & $\psi_{k1}$ & no & 2.6s & 850 & \multicolumn{2}{c|}{ }  & 343&327    \\ 
  &   
  &  & $\psi_{k2}$ & no & 2.4s & 875 & \multicolumn{2}{c|}{ }  & 336&320    \\
\hline
(l) & auction           
  & & no deadlock & no & 10.8s & 1683 & 5&7 & 186&206 \\
  &   
  &  & $\psi_{l1}$ & no & 6.4s & 1180 & \multicolumn{2}{c|}{ }  & 79&87 \\
  && & $\psi_{l2}$ & yes & 26.5s & 4000 & \multicolumn{2}{c|}{ }  & 263&378
\end{tabular}
\end{center}

We briefly comment on the benchmarks:
For all examples we checked the property \emph{no deadlock} that abbreviates $\A \G \E \F \chi_f$,
where $\chi_f$ is a disjunction of all final states. This is one of the two requirements
of the crucial \emph{soundness} property (cf. \exaref{road fines}).
Weak soundness~\cite{BatoulisHW17} that relaxetion that allows dead transitions, but
all firable transitions must lead to final states. We write
\emph{weak sound(a)} for the property
$\E\F (\langle a\rangle \top) \to \A\G (\langle a\rangle \top \to \F \chi_f)$,
stating the requirements for action $a$.
\begin{compactdesc}
\item[(a)-(c)] are versions of the road fine process from \exaref{road fines}.
The DPNs for (a)~\cite[Fig. 12.7]{Mannhardt18} and (b)~\cite[Fig. 13]{LFM18} were mined automatically from logs, while (c) is the normative version \cite[Fig. 7]{MannhardtLRA16} shown in \exaref{road fines}.
While (a) and (c) are unsound (\emph{no deadlock} is violated), this issue
was fixed in version (b).
We can also check whether specific states are deadlock-free, as by
$\psi_{a1} = \A \G (\mathsf{p}_{7} \to \E\F \mathsf{end})$, which actually holds in (a)-(c) 
as $\mathsf{p}_{7}$ is not the problematic state.
Other considered properties are
$\psi_{a2} = \A \G (\mathsf{end} \to \mathit{total}\,{\leq}\, \mathit{amount})$, which states that in the final state it is ensured that the total amount exceeds the fine.
Moreover, $\psi_{c1} = \E \F (\mathit{dS}\,{\geq}\, 2160)$, 
$\psi_{c2} = \E \F (\mathit{dP}\,{\geq}\, 1440)$, and
$\psi_{c3} = \E \F (\mathit{dJ}\,{\geq}\, 1440)$ check whether the time constraints can be violated.
\item[(d)] models a billing process in a hospital \cite[Fig. 15.3]{Mannhardt18}. The tool verifies that it is deadlock-free.
Moreover, $\psi_{d1} = \E \F (\mathsf{p16} \wedge \neg \mathit{isClosed})$ checks whether there exists a run where in the final state $\m{p16}$ the $\mathit{isClosed}$ flag is not set.
\item[(e)] is a normative model for a sepsis triage process in a hospital \cite[Fig. 13.3]{Mannhardt18}, and
(f) is a version of the same process that was mined purely automatically from logs \cite[Fig. 13.6]{Mannhardt18}.
Both versions are deadlock-free.
According to \cite[Sec. 13]{Mannhardt18}, it is assumed that triage happened before antibiotics are administered, i.e.,
$\psi_{e1} =\A \G (\mathsf{sink} \to \mathit{timeTriage} < \mathit{timeAntibiotics})$, 
which is actually not satisfied by (e). However, the desired time limit
$\psi_{e2} =\A \G (\mathsf{sink} \to \mathit{timeTriage +60} \geq \mathit{timeAntibiotics})$ holds.
We can check that variable $\mathsf{lacticAcid}$ is not written until a certain activity happens, i.e.,
$\psi_{f1} =\A (\neg \mathsf{lacticAcid} \U \langle\m{diagnosticLacticAcid}\rangle \top)$ holds.
\item[(g)--(i)] reflect activities in patient logistics of a hospital, based on logs of real-life processes \cite[Fig. 14.3]{Mannhardt18}.
While the \emph{no deadlock} property is satisfied by all initial configurations,
the output of \tool reveals that in case of (h) this need not hold for other initial assignments.
The tool also confirms that if the variable $\mathit{org}_1$ has value 207 in state $\m{p}_2$ then
this value will be maintained, $\psi_{i1} = \A \G ( \m{p_2} \wedge \mathit{org}_1 \eqn 207 \to \A \G \mathit{org}_1 \eqn 207)$.
We also verify that in this process either the $\m{transfer}$ or $\m{history}$ activity happens,
but not both, by $\psi_{i2} = \A ( \E \F \langle\m{transfer}\rangle \top \wedge \E \F \langle\m{history}\rangle\top)$ and $\psi_{i3} = \neg \E (\F \langle\m{transfer}\rangle \top \wedge \F \langle\m{history}\rangle\top)$.
\item[(j)] is a credit approval process \cite[Fig. 3]{LM18}.
It can be verified that a loan is only granted if the application passed the 
customer verification and the decision stages ($\psi_{j1} = \A \G (\langle\m{openLoan}\rangle \top \to \mathit{ver} \wedge \mathit{dec})$); though even if the verification and the decision variables are set, it is not guaranteed that a loan is granted ($\psi_{j2} = \A (\F (\mathit{ver} \wedge \mathit{dec}) \to \F (\langle\m{openLoan}\rangle \top))$), but it is possible
($\psi_{j3} = \A (\F (\mathit{ver} \wedge \mathit{dec}) \to \E \F (\langle\m{openLoan}\rangle \top))$)
\item[(k)] is a package handling routine \cite[Fig. 5]{deLeoniFM21}. The properties $\psi_{k1} = \E \F \langle\m{fetch}\rangle \top$ and
$\psi_{k2} = \E \F \langle\mathit{\tau_6}\rangle \top$ are not satisfied, so the process has dead transitions.
\item[(l)]models an auction process \cite[Ex. 1.1]{ada}, for which \tool reveals a deadlock.
We also check the properties $\psi_{l1} = \E\F(\mathsf{sold} \wedge d\,{>}\,0 \wedge o\,{\leq}\,t)$ and
$\psi_{l2} = \E \F(b = 1 \wedge o\,{>}\,t \wedge \F (\mathsf{sold} \wedge b\,{>}\,1))$ considered in \cite[Ex. 1.1]{ada}.
\end{compactdesc}

\section{Conclusion}
This paper presents a \CTLsf verification technique for DDSAs that is a decision procedure for monotonicity and integer periodicity constraint systems.
To the best of our knowledge, this is the first proof of decidability of \CTLsf for these classes. In contrast, the cases of feedback-free and
bounded lookback systems are shown undecidable.
We implemented our approach in the tool \tool
and showed its usefulness on a range of business processes from the literature.

We see various opportunities to extend this work.
A richer verification language could support past time
operators~\cite{Demri06} and the possibility to compare variables multiple steps apart~\cite{DD07}. Further decidable fragments could be sought using covers~\cite{GM08}, or aiming for
compatibility with locally finite theories~\cite{GNRZ07}.
Moreover, a restricted version of the bounded lookback property could guarantee
decidability of \CTLsf, similarly to the way feedback freedom was strengthened in ~\cite{KoutsosV17}. We conjecture that many of the DPNs used in the experiments
could be in such a class.
The implementation could be improved to avoid the computation of many similar formulas,
thus gaining efficiency.
Finally, the complexity class that our approach implies for \CTLsf in the decidable classes
is yet to be clarified.

\bibliographystyle{splncs04}
\bibliography{references}

\appendix
\section{Proofs}
\label{app:proofs}

\begin{numberedlemma}{\ref{lem:abstraction}}
For a symbolic run $\sigma\colon b_0 \goto{a_1} b_1 \goto{a_2} \dots \goto{a_n} b_n$ and $\vseq= \langle \theta_0,\dots, \theta_n\rangle$,
$\smash[t]{\hist(\sigma,\vseq)}$
is satisfied by assignment $\alpha$ with domain $V\,{\cup}\,V_0$ iff 
$\sigma$ abstracts a run
$\smash{\rho \colon (b_0, \alpha_0) \goto{a_1} \dots \goto{a_n} (b_n, \alpha_n)}$
such that 
\begin{inparaenum}[\it (i)]
\item $\alpha_0(v) = \alpha(\nu(v))$, and
\item $\alpha_n(v) = \alpha(v)$ for all $v\in V$, and
\item $\alpha_i \models \theta_i$ for all $i$, $0\leq i \leq n$.
\end{inparaenum}
\end{numberedlemma}
\begin{proof}
($\Longleftarrow$)
By induction on $n$.
If $n\,{=}\,0$, the assumptions imply $\alpha(v) = \alpha(\nu(v))$ for all $v\in V$, so $\alpha \models \phi_\nu$.
As $\alpha_0$ satisfies $\theta_0$, 
$\alpha$ also satisfies 
$ \phi_{\nu} \wedge \theta_0 = \hist(\sigma, \langle \theta_0\rangle)$, so the claim holds.
For the induction step, let $\sigma$ be a symbolic run 
$\sigma \colon b_0 \goto{*} b_n \goto{a} b_{n+1}$
that abstracts a run
$\rho \colon
(b_0, \alpha_0) 
\goto{*} (b_n, \alpha_n)
\goto{a} (b_{n+1}, \alpha_{n+1})
$, and let
$\vseq\,{=}\, \langle \theta_0,\dots, \theta_n, \theta_{n+1}\rangle$.
Then $\sigma|_n$ also abstracts $\rho|_n$, 
so by the induction hypothesis the assignment $\alpha'$ with domain $V \cup V_0$ given by
$\alpha'(\nu(v)) = \alpha_0(v)$ and $\alpha'(v) = \alpha_n(v)$ for all $v\in V$ 
satisfies $\hist(\sigma|_n, \vseq|_n)$.
By definition of a step,
the guard assignment $\beta$ given by $\beta(v^r) = \alpha_n(v)$ and
$\beta(v^w) = \alpha_{n+1}(v)$ for all $v\in V$
satisfies $\guard(a)$.
For $\varphi := \hist(\sigma|_n, \vseq|_n)$ we thus have
\begin{align*}
\hist(\sigma, \vseq) &= \update(\varphi, a) \wedge \theta_{n+1}\\
&= \exists \vec U.\:\varphi(\vec U) \wedge \trans{a}(\vec U, \vec V) \wedge \theta_{n+1}\\
&= \exists \vec U.\:\varphi(\vec U) \wedge
(\guard(a) \wedge\!\bigwedge_{v\in V \setminus \mathit{write}(a)}\! v^{w}\,{=}\,v^{r})(\vec U, \vec V) \wedge \theta_{n+1}
\end{align*}
As $\alpha' \models \phi$, by the construction of $\beta$
above, it holds that $\alpha_{n+1}$ satisfies the first conjunct of this formula,
using values $\alpha_n(\vec V)$ as
witnesses for the existentially quantified variables $\vec U$.
Since moreover $\alpha_{n+1}\models \theta_{n+1}$ by assumption, it follows that
$\alpha_{n+1}$ satisfies $\hist(\sigma, \vseq)$.

($\Longrightarrow$)
By induction on $n$.
For $n\,{=}\,0$, suppose that $\alpha$ satisfies
$\hist(\sigma,\vseq) = \phi_{\nu} \wedge \theta_0$.
By definition of $\phi_\nu$, this implies $\alpha(v) = \alpha(\nu(v))$ for all $v\in V$.
The empty run $(b_0,\alpha_0)$ thus satisfies the claim, with $\alpha_0(v) = \alpha(v)$ for all $v\in V$.
For the inductive step, let
$\sigma \colon b_0 \goto{*} b_n \goto{a} b_{n+1}$ 
and $\vseq= \langle \theta_0,\dots, \theta_{n+1}\rangle$
satisfy
$\alpha \models \hist(\sigma, \vseq)$. 
Since 
\begin{align*}
\hist(\sigma, \vseq)
&= \update(\hist(\sigma|_n, \vseq|_n), a) \wedge \theta_{n+1} \\
&= \exists \vec U. \hist(\sigma|_n, \vseq|_n)(\vec U) \wedge \trans{a}(\vec U, \vec V) \wedge \theta_{n+1} 
\end{align*}
it must hold that $\alpha \models \theta_{n+1}$ and
there must be an assignment $\gamma$ with domain $U \cup V_0 \cup V$ such that $\gamma(v) = \alpha(v)$ for all $v\in V \cup V_0$, and $\gamma$ satisfies both
$\hist(\sigma|_n, \vseq|_n)(\vec U)$ and $\trans{a}(\vec U, \vec V)$.
We can write $\vec V = \langle v_1, \dots, v_k\rangle$ and 
$\vec U = \langle u_1, \dots, u_k\rangle$ for some $k$.
Let $\alpha'$ be the assignment with domain $V \cup V_0$ such that $\alpha'(v_i) = \gamma(u_i)$ for all $i$, $1\leq i \leq k$, and $\alpha'(v) = \gamma(v)$ for all $v\in V_0$.
Then $\alpha'$ satisfies $\hist(\sigma|_n, \vseq|_n)$.
Therefore, by the induction hypothesis $\sigma|_n$ abstracts a run
$\rho \colon
(b_0, \alpha_0) 
\goto{a_1} (b_1, \alpha_1)
\goto{a_2} \dots
\goto{a_n} (b_n, \alpha_n)$ such that 
$\alpha_i \models \theta_i$ for all $i$, $0\,{\leq}\,i\,{\leq}\,n$.
Let $\beta$ be the guard assignment such that $\beta(v^r) = \alpha'(v)$ and
$\beta(\vec v^w) = \alpha(v)$ for all $v\in V$.
By definition of $\alpha'$, since $\gamma$ satisfies $\trans{a}(\vec U, \vec V)$, $\beta$ satisfies $\trans{a}(\vec V^r, \vec V^w)$ and hence $\beta \models \guard(a)$.
Thus $\rho$ can be extended with a step $(b_n, \alpha_n)\goto{a} (b_{n+1}, \alpha_{n+1})$
such that $\alpha_{n+1}(v) = \alpha(v)$ for all $v\in V$.
Moreover, as $\alpha$ satisfies $\theta_{n+1}$,
$\alpha_i \models \theta_i$ for all $i$, $0\,{\leq}\,i\,{\leq}\,{n+1}$.
This proves the claim.
\qed
\end{proof}

Given a path $\pi$ in $\NN_{\BB,b}^\psi$ of the form
\begin{equation}
\label{eq:pi}
\pi\colon
(\bdummy,q_0,\phi_\nu)
\goto{a_0,w_0} (b_0, q_1,\varphi_1)
\goto{a_1,w_1} (b_1, q_2,\varphi_2)
\goto{*} (b_n, q_{n+1},\varphi_{n+1})
\end{equation}
where the last node is final,
we write $\sigma(\pi)$ for the symbolic run
$\sigma\colon b = b_0 \goto{a_1} b_1 \goto{*} b_n$
(ignoring the initial dummy transition in $\pi$), and
$w(\pi) = \seqz[n]{w}$. 

\begin{lemma}
\label{lem:pc}
Let $\psi \in \LTLconf$ be over $\Phi$.
\begin{compactenum}
\item
If a word $w$ is accepted by $\NN_\psi$ and $\sigma$ is a symbolic run
such that
$\hist(\sigma,w\otimes \sigma)$ is satisfiable, 
there is a path $\pi$ of the form \eqref{pi}
in $\smash{\NN_{\BB,b}^\psi}$ such that $\sigma = \sigma(\pi)$, $w = w(\pi)$, 
and $\varphi_{n+1} \equiv \hist(\sigma,w\otimes \sigma)$.
\item
If $\pi$ is a path of the form \eqref{pi}
in $\NN_{\BB,b}^\psi$ then $w(\pi)$ is accepted by $\NN_\psi$,
$\varphi_{n+1}$ is satisfiable, and
$\varphi_{n+1} \equiv h(\sigma(\pi),w(\pi) \otimes \sigma(\pi))$.
\end{compactenum}
\end{lemma}
\begin{proof}
\begin{compactenum}[(1)]
\item 
By induction on the length $n$ of $\sigma$.
If $n\,{=}\,0$ then $\sigma$ is empty and
$w = \varsigma$ for some $\varsigma \in \Sigma$.
By assumption,
$\hist(\sigma,w \otimes \sigma) = \phi_\nu \wedge \varsigma_0(b)$ is satisfiable. 
Thus, 
$\update(\phi_\nu,a_0) \wedge \varsigma_0(b) = \phi_\nu \wedge \varsigma_0(b)$
is satisfiable (using $\guard(a_0) = \top$), so
by \defref{product construction} there is a step
$(\widehat b, q_0, \phi_\nu) \goto{a_0} (b, q_f, \varphi_1)$ such that
$\varphi_1 \equiv \update(\phi_\nu,a_0) \wedge \varsigma_0(b)$.

In the inductive step, $\sigma$ has the form $b_0 \goto{*} b_n \goto{a} b_{n+1}$, and 
$w = \varsigma_0 \cdots \varsigma_{n+1}$
is accepted by $\NN_\psi$, such that 
$\hist(\sigma,w \otimes \sigma)$ is satisfiable.
Let $\sigma' = \sigma|_n$ and $w' = w|_n$.
By the induction hypothesis, $\NN_\BB^\psi$ has a 
node $p_{n+1} = (b_n, q_{n+1},\varphi_{n+1})$
and a path 
$\pi \colon p_0 \to^* p_{n+1}$
such that $\varphi_{n+1} \equiv \hist(\sigma',w'\otimes\sigma')$.
Therefore,
\begin{align*}
\update(\varphi_{n+1}, a) \wedge \varsigma_{n+1}(b_{n+1}) &\equiv \update(\hist (\sigma',w'\otimes\sigma'), a) \wedge \varsigma_{n+1}(b_{n+1}) \\
&= \hist(\sigma,w\otimes\sigma)
\end{align*}
is satisfiable.
Therefore, $\NN_\BB^\psi$ must have a node $p' = (b_{n+1}, q_{n+2}, \varphi_{n+2})$ 
such that $\varphi_{n+2} \equiv \update(\varphi_{n+1}, a) \wedge \varsigma_{n+1}(b_{n+1})$
and an edge $p_n \goto{a, \varsigma_{n+1}} p'$ can be appended to $\pi$.
\item
By induction on $n$.
If $n\,{=}\,0$ then $\pi$ consists of the single step
$(\widehat b,q_0,C_{\alpha_0}) \goto{a_0} (b,q_1,\varphi_1)$
and $\sigma$ consists only of state $b$. 
By \defref{product construction},
this step exists because 
$\update(\phi_\nu, a_0) \wedge \varsigma_0(b) = \phi_\nu \wedge \varsigma_0(b)$ is satisfiable,
for some $q_0 \goto{\varsigma_0} q_1$, using the fact that $\guard(a_0) = \top$.
The formula $\varphi_1$ must satisfy
$\varphi_1 \equiv \phi_\nu \wedge \varsigma_0(b)$.
For $w(\pi) = \varsigma_0$ we indeed have 
$\hist(\sigma,w(\pi) \otimes \sigma(\pi)) = \phi_\nu \wedge \varsigma_0(b)$,
so the claim holds.

In the inductive step, consider a path
$\pi \colon p_0 \to^* p_{n+1} \goto{a} p_{n+2}$
for $p_0$ the initial node of $\smash{\NN_\BB^\psi}$
and $p_i = (b_{i-1}, q_{i},\varphi_{i})$ for all $i$, $1\,{\leq}\,i\,{\leq}\,n\,{+}\,2$.
Let $\sigma = \sigma(\pi)$ 
be the symbolic run $b_0 \goto{*} b_{n} \goto{a} b_{n+1}$, $\sigma' = \sigma|_n$,
and $w = w(\pi)$. 
By the induction hypothesis, there is a run
$q_0 \goto{\varsigma_0} q_1 \goto{\varsigma_2} \dots \goto{\varsigma_n} q_{n+1}$ 
in $\NN_\psi$ such that for
$w' = w(\pi|_n) = \varsigma_0 \dots \varsigma_n$, the history constraint
$\hist(\sigma',w' \otimes \sigma')$ is satisfiable and equivalent to
$\varphi_{n+1}$ ($\star$).
Since there is an edge 
$(b_n, q_{n+1},\varphi_{n+1}) \goto{a} (b_{n+1}, q_{n+2},\varphi_{n+2})$, 
by \defref{product construction} there must be a transition
$q_{n+1} \goto{\varsigma} q_{n+2}$ in $\NN_\psi$, 
such that $\varphi_{n+2} \equiv \update(\varphi_{n+1}, a) \wedge \varsigma(b_{n+1})$ is satisfiable.
Using ($\star$) and abbreviating $\theta = \varsigma(b_{n+1})$, 
\begin{align*}
\update(\varphi_{n+1}, a) \wedge \theta
& \equiv \update(\hist(\sigma',w' \otimes \sigma'), a) \wedge \theta
\eqn\hist(\sigma,w\,{\otimes}\,\sigma)
\end{align*}
holds and since
$\varphi_{n+2}$ is satisfiable the claim holds.
\qed
\end{compactenum}
\end{proof}

For instance, the path to the left final node in $\smash{\NN^\psi_{\BB,\m b_1}}$ in \exaref{pc} corresponds to the word $w=\langle \emptyset,\{K\}\rangle$
accepted by $\NN_\psi$ and $\sigma\colon \m b_1 \goto{\m a_1} \m b_2 \goto{\m a_3} \m b_3$, and the formula is equivalent to $\hist(\sigma, \vec\theta)$
for $\vec \theta = \langle \top, K(\m b_2), \top\rangle$.
Now, the product construction serves to check
whether there exists a run that satisfies an $\LTLconf$ formula:

\begin{numberedtheorem}{\ref{thm:model:checking}}
Let $\psi \inn \LTLconf$ and $b\inn B$.
There is a final run $\rho \colon (b,\alpha_0) \to^* (b_F, \alpha_F)$ of $\BB$ 
such that $\rho \modelsLTL \psi$, iff
$\smash{\NN_{\BB,b}^\psi}$ has a final state $(b_F, q_F, \phi)$ for some $q_F$ and $\phi$
such that $\phi$ is satisfied by assignment $\gamma$ with
$\gamma(\vec {V_0}) \eqn \alpha_0(\vec V)$ and $\gamma(\vec V) \eqn \alpha_F(\vec V)$.
\end{numberedtheorem}
\begin{proof}
($\Longrightarrow$)
Suppose $\rho \modelsLTL \psi$, and let $\sigma$ be the abstraction of $\rho$.
By \lemref{nfa}, $\NN_\psi$ accepts a word $w = w_0 \dots w_n$ that is consistent with $\rho$, i.e., $\alpha_i \models w_i(b_i)$ 
for all $i$, $0\,{\leq}\,i\,{\leq}\,n$. Thus assignment
$\gamma$ satisfies $\hist(\sigma,w\otimes \sigma)$ by \lemref{abstraction}.
By \lemref{pc}, there is a path $\pi$ in $\smash{\NN_{\BB,b}^\psi}$
ending in a state $(b_F, q_F,\phi)$ such that $\sigma = \sigma(\pi)$, $w = w(\pi)$,
and $\phi \equiv \hist(\sigma,w\otimes \sigma)$.
($\Longleftarrow$) 
Let $\pi$ be a path to a final state $(b_F, q_F,\phi)$ in $\smash{\NN_{\BB,b}^\psi}$.
By \lemref{pc}, $w(\pi)$
is accepted by $\NN_\psi$, and $\hist(\sigma(\pi),w(\pi) \otimes \sigma(\pi))$ is equivalent to $\phi$ 
and satisfied by some assignment $\gamma$. 
By \lemref{abstraction}, there is a run 
$\rho\colon (b, \alpha_0) \to^* (b_F, \alpha_n)$
with abstraction $\sigma$
such that $\gamma(\vec V_0) = \alpha_0(\vec V)$, and $\gamma(\vec V)=\alpha_n(\vec V)$ and 
$\alpha_i  \models w_i(b_i)$ for all $i$, $0\,{\leq}\,i\,{\leq}\,n$. So $w$ is consistent with 
$\rho$, and by \lemref{nfa} we have $\rho \modelsLTL \psi$.
\qed
\end{proof}

\begin{numberedlemma}{\ref{lem:toLTL}}
Let $\psi$ be a path formula with $\qd(\psi) \eqn k$.
Suppose that for all confi\-gurations $(b,\alpha)$ and path formulas $\psi'$ with $\qd(\psi')\,{<}\,k$, there is a $\rho' \in \FRuns(b,\alpha)$ with $\rho',0 \models \psi'$ iff $\alpha \models \chP(\psi')(b)$. Then
$\rho,0 \models \psi$ iff $\rho,0 \modelsLTL \toLTLconf(\psi)$.
\end{numberedlemma}
\begin{proof}
We suppose that $\rho_0 = (b_0, \alpha_0)$ and apply induction on $\psi$.
First, if $\psi=\top$ then $\rho,0 \modelsLTL \top$ and $\toLTLconf(\psi) = \top$, so the claim holds.
Second, if $\psi = b \in B$ then $\rho,0 \models \psi$ iff $b=b_0$,
and moreover $\rho,0 \modelsLTL \toLTLconf(\psi) = K_b$ iff $\alpha_0 \models K_b(b_0)$,
which holds iff $b=b_0$.
Third if $\psi = c$, then $\rho,0 \models \psi$ iff $\alpha_0 \models c$,
and $\rho,0 \modelsLTL \toLTLconf(\psi) = K_c$ iff $\alpha_0 \models K_c(b_0) = c$.
For the induction step, we perform again a case distinction on $\psi$.
If $\psi = \E \psi'$ then $\rho,0 \models \psi$ iff there is some $\rho' \inn \FRuns(b_0, \alpha_0)$ with $\rho',0 \models \psi'$. As $\qd(\psi') < \qd(\psi)$, this holds by assumption iff
$\alpha_0 \models \chP(\psi')(b_0)$.
Moreover, $\rho,0 \modelsLTL \toLTLconf(\psi) = \chP(\psi')$ iff
$\alpha_0 \models \chP(\psi')(b_0)$ by definition of $\modelsLTL$, which proves the claim.
All remaining cases follow from the induction hypothesis and the fact that the definitions of $\LTLconf$ semantics (\defref{ltl:semantics}) and \CTLsf path semantics (\defref{semantics}) coincide in their recursive structure for the boolean and temporal operators.
\qed
\end{proof}

\begin{numberedtheorem}{\ref{thm:main}}
 For every configuration $(b,\alpha)$ of $\BB$ and state property $\chi$,
$(b,\alpha) \models \chi$ iff $\alpha \models \chS(\chi)(b)$.
\end{numberedtheorem}
\begin{proof}
We first show property ($\star$): there is a run $\rho \in \FRuns(b,\alpha)$ with $\rho,0 \models \psi$ iff $\alpha \models \chP(\psi)(b)$.
The proof is by induction on $\qd(\psi)$.
If $\psi$ contains no path quantifiers, \lemref{toLTL} implies that
$\rho,0 \models \psi$ iff $\rho,0 \modelsLTL \toLTLconf(\psi)$ for all runs $\rho$, so ($\star$)
follows from \lemref{chP}.
In the induction step, we conclude from \lemref{toLTL} (a), using the induction hypothesis as assumption, that $\rho,0 \models \psi$ iff $\rho,0 \modelsLTL \toLTLconf(\psi)$ for all runs $\rho$. Then ($\star$) follows from \lemref{chP}.

The claim of the lemma can be shown by induction on $\chi$:
There are three base cases:
if $\chi\eqn\top$ or $\chi \eqn c$ then the claim is trivial as $\chS(\chi)(b)=\chi$; and
if $\chi = b'$ for some $b'\in B$, $(b,\alpha) \models \chi$ iff $b=b'$,
and the same condition applies to $\alpha \models \chS(\chi)(b)$ by definition of $K_{b'}$.
The inductive step also distinguishes three cases:
First, if $\chi = \neg \chi'$ then by the induction hypothesis
$(b,\alpha) \models \chi'$ iff $\alpha \models \chS(\chi')(b)$.
So by definition, $(b,\alpha) \models \chi$ iff $(b,\alpha) \not\models \chi'$ iff
$\alpha \not\models \chS(\chi')(b)$, which holds by definition of negation of configuration maps iff $\alpha \models \chS(\chi)(b)$.
Second, if $\chi = \chi_1 \wedge \chi_2$ then by the induction hypothesis
$(b,\alpha) \models \chi_i$ iff $\alpha \models \chS(\chi_i)(b)$ for both $i\in \{1,2\}$.
So by definition, $(b,\alpha) \models \chi$ iff both $(b,\alpha) \models \chi_i$, which holds iff
$\alpha \models \chS(\chi_i)(b)$ for both $i\in \{1,2\}$. By definition of conjunction on configuration maps this is equivalent to $\alpha \models \chS(\chi)(b)$.
Finally, if $\chi = \E \psi$ then $(b,\alpha) \models \chi$ iff there is a run $\rho \in \FRuns(b,\alpha)$
such that $\rho, 0 \models \psi$. We use ($\star$) to conclude that this is the case iff 
$\alpha \models \chP(\psi)(b) = \chS(\chi)(b)$.
\qed
\end{proof}

\begin{numberedlemma}{\ref{lem:formula:sets}}
Let $\psi$ have quantifier depth $k$, $\psi' = \toLTLconf(\psi)$, and 
$\smash{\NN_{\BB,b}^{\psi'}}$ be a constraint graph constructed in $\chP(\psi)$ for some $b\in B$.
Then, 
\begin{compactenum}[(1)]
\item for all nodes $(b', q, \phi)$ in $\smash{\NN_{\BB,b}^{\psi'}}$ there is some $\phi' \in \HH_k$ such that $\phi \equiv \phi'$, 
\item $\chP(\psi)(b)$ is equivalent to a formula in $\Phi_{k+1}$.
\end{compactenum}
\end{numberedlemma}
\begin{proof}
We prove the statements by induction on $k$. In the base case, $\psi$ contains
no path quantifiers, so by definition of $\toLTLconf$,
all atoms $K'$ occurring in $\psi'$ satisfy $K'(b') \in \CC\cup\{\top, \bot\} = \Phi_0$ for all $b' \in B$, so $\psi'$ is an $\LTLconf$ formula over basis $\Phi_0$. 
Let $\pi$ be a path to a node $(b', q, \phi)$ in $\smash{\NN_{\BB,b}^{\psi'}}$,
$\sigma:= \sigma(\pi)$ be the associated symbolic run $\sigma \colon b_0 \goto{a_1} \dots \goto{a_n} b_n$, and
$w:=w(\pi)$ the associated word in $\Sigma^*$.
By \lemref{pc}\,(1), $\varphi \equiv h(\sigma,w\otimes \sigma)$.
The word $w = \seqz{w}$ satisfies $w_i\in 2^{\Confs(\Phi_0^\pm)}$, so for
$w\otimes \sigma = \seqz{\theta}$ we have $\theta_i = \bigwedge T_i$ 
for some $T_i \subseteq \Phi_0^\pm$.
Therefore
$h(\sigma,w\otimes \sigma) \in \HH_0$, so that (1) holds.
Furthermore,
the symbolic configuration map $\Conf$ returned by $\chP$ satisfies
$\Conf(b) = \bigvee_{\phi \in \Phi_F} \exists \vec U.\ \phi(\vec V,\vec U)$,
where every $\phi$ is equivalent to some formula in 
$\HH_0$. Hence $\Conf(b)$ is equivalent to a formula in $\Phi_1$ by definition of $\Phi$, so that $(2)$ holds.

In the step case, we consider a formula $\psi$ of quantifier depth $k+1$, and
the induction hypothesis is that $(1)$ and $(2)$ hold for a formula of depth $k$.
The call to $\toLTLconf(\psi)$ replaces all occurrences of subformulas 
$\E \eta$ by
$\chP(\eta)$, where $\eta$ has quantifier depth at most $k$.
By part (2) of the induction hypothesis, $\chP(\eta)(b)$ is equivalent to a formula in $\bigcup_{i=0}^{k+1}\Phi_{k+1}$ for all $b\in B$.
If we abbreviate $\Theta := \bigcup_{i=0}^{k+1}\Phi_{k+1}^\pm$,
we can thus assume that 
$\psi' = \toLTLconf(\psi)$ is a formula over $\Theta$.
Let $\pi$ be a path to a node $(b', q, \phi)$ in $\smash{\NN_{\BB,b}^{\psi'}}$,
$\sigma:= \sigma(\pi)$ be the associated symbolic run $\sigma \colon b_0 \goto{a_1} \dots \goto{a_n} b_n$, and $w:=w(\pi)$ the associated word.
By \lemref{pc}\,(1), $\varphi \equiv h(\sigma,w\otimes \sigma)$.
The word $w = \seqz{w}$ satisfies $w_i\in 2^{\Confs(\Theta)}$, so for
$w\otimes \sigma = \seqz{\theta}$ we have $\theta_i  = \bigwedge T_i$ for some $T_i \subseteq \Theta$.
Therefore
$h(\sigma,w\otimes \sigma) \in \HH_{k+1}$, so that (1) holds.
Furthermore,
the symbolic configuration map $\Conf$ returned by $\chP$ satisfies
$\Conf(b) = \bigvee_{\phi \in \Phi_F} \exists \vec U \phi(\vec V,\vec U)$,
where every $\phi$ is equivalent to some formula in 
$\HH_{k+1}$.
Hence $\Conf(b)$ is equivalent to a formula in $\Phi_{k+2}$,
so that $(2)$ holds.
\qed
\end{proof}

In IPCs of the form $x \odot d$ for $\odot \in \{=,\neq, <, >\}$, $x \equiv_k y + d$, and $x \equiv_k d$, we call $d$ a \emph{constant} and $k$ a \emph{modulus}.

\begin{numberedtheorem}{\ref{thm:ipc}}
For any IPC-DDSA $\BB$ and IPC-property $\chi$ the verification problem is decidable.
\end{numberedtheorem}
\begin{proof}
\newcommand{\IPCC}{\Phi_{\mathit{IPC}}}
Let $\chi$ be an IPC-property, $\Const$ the finite set of constants $d$ in $\chi$, ${\alpha_0}$, and guards of $\BB$, and $K$ the least common multiple of all moduli $k_1, \dots, k_m$ that occur in $\chi$ and guards of $\BB$.
Let moreover $\IPCC$ be the set of quantifier-free boolean formulas whose atoms are IPCs over variables $V \cup V_0$, moduli $k_1, \dots, k_m, K$, and constants $\Const$,
so $\IPCC$ is finite up to equivalence.

We show the following property ($\star$): all history constraints $\hist(\sigma,\vec \theta)$ 
over $\IPCC$ are equivalent to a formula in $\IPCC$.
For a symbolic run 
$\sigma\colon b_0 \goto{a_1} b_1 \goto{a_2} \dots \goto{a_n} b_n$ and a sequence
$\vseq= \langle \theta_0,\dots, \theta_n\rangle$ over $\IPCC$, the proof is by induction on $n$.
In the base case $n=0$, $\hist(\sigma, \vseq) \eqn \phi_\nu \wedge \theta_0$ is in $\IPCC$ because
$\phi_\nu$ is a conjunction of equalities between variables, and $\theta_0 \in \IPCC$ by assumption.
In the induction step, $\hist(\sigma, \vseq) \eqn \update(\hist(\sigma|_{n-1}, \vseq|_{n-1}), a_{n}) \wedge \theta_n$. By induction hypothesis, $\hist(\sigma|_{n-1}, \vseq|_{n-1})$ is equivalent to a formula $\phi$ in $\IPCC$.
Thus $\hist(\sigma, \vseq) \equiv \phi $ for the formula 
$\phi = \exists \vec U. \phi(\vec U) \wedge \Delta_{a_n}(\vec U, \vec V) \wedge \theta_n$.
By assumption, $\Delta_a(\vec U, \vec V)$ is a
conjunction of IPCs over $V\cup U$, moduli $k_1, \dots, k_m$, and $\Const$, and $\theta_n \in \IPCC$ as well.
According to the quantifier elimination property proven in \cite[Thm. 2]{Demri06}, there exists a quantifier-free IPC-formula $\phi'$ over variables $V_0\cup V$, modulus $K$, and $\Const$ that is equivalent to $\exists \vec U. \phi(\vec U) \wedge \Delta_{a_n}(\vec U, \vec V) \wedge \theta_n$, so $\phi' \in \IPCC$.

We now show that $\HH_j(b) \subseteq \IPCC$ for all $j\geq 0$, by induction on $j$.
In the base case ($j\eqn 0$) the claim follows from ($\star$), since all constraints in $\chi$ are in $\IPCC$.
For $j>0$, consider first a formula $\widehat \phi \in\Phi_{j}$.
Then $\widehat \phi$ is of the form
$\widehat \phi = \bigvee_{\phi \in H} \exists \vec U.\ \phi(\vec V, \vec U)$ for some
$H\,{\subseteq}\, \HH_{j-1}$. By the induction hypothesis, $H\subseteq \IPCC$,
so by the above quantifier elimination property, $\widehat \phi$ is equivalent
to a formula $\phi' \in \IPCC$. As $\HH_j$ consists of all history constraints over $\Phi_{j}$, the
claim follows from $(\star)$.
\qed
\end{proof}

\section{Examples}
\label{app:examples}

\begin{example}
\label{exa:mc:continued}
We show here the product automata that were omitted in \exaref{mc} for lack
of space.
For the $\LTLconf$ formula $\psi_2' = \F (\Conf_{x<2})$ and state $\m b_1$, we get the following automaton:
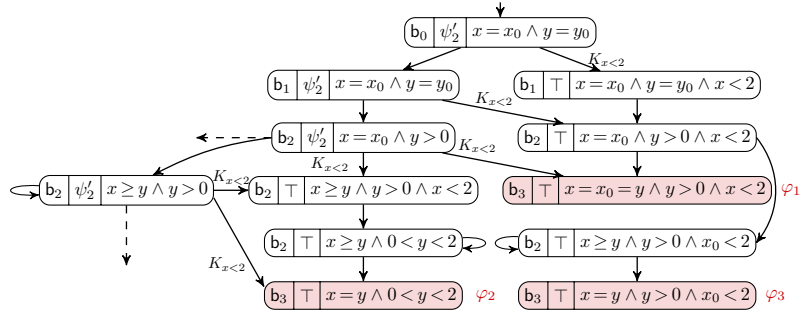
\begin{center}
\begin{tikzpicture}[node distance=10mm,>=stealth']
\tikzstyle{node} = [draw,rectangle split, rectangle split parts=3,rectangle split horizontal, rectangle split draw splits=true, inner sep=3pt, scale=.7, rounded corners]
\tikzstyle{goto} = [->]
\tikzstyle{action}=[scale=.6, black]
\tikzstyle{constr}=[scale=.5, black]
\tikzstyle{accepting state} = [fill=red!80!black!15]
\node[node] (0)  {\pcnode{$\m b_0$}{$\psi_2'$}{$x\eqn x_0 \wedge y \eqn y_0$}};
\node[node, below of = 0, xshift=-26mm] (1a)
  {\pcnode{$\m b_1$}{$\psi_2'$}{$x\eqn x_0 \wedge y \eqn y_0 $}};
\node[node, below of = 0, xshift=26mm] (1b)
  {\pcnode{$\m b_1$}{$\top$}{$x\eqn x_0 \wedge y \eqn y_0 \wedge x\,{<}\,2$}};
\node[node, below of = 1a] (2a)
  {\pcnode{$\m b_2$}{$\psi_2'$}{$x\eqn x_0 \wedge y\,{>}\,0$}};
\node[node, below of = 1b] (2b)
  {\pcnode{$\m b_2$}{$\top$}{$x\eqn x_0 \wedge y\,{>}\,0 \wedge x\,{<}\,2$}};
\node[node, below of = 2b, accepting state] (3b)  
  {\pcnode{$\m b_3$}{$\top$}{$x\eqn x_0 \eqn y \wedge y\,{>}\,0 \wedge x\,{<}\,2$}};
\node[node, below of = 2a] (3a)  
  {\pcnode{$\m b_2$}{$\top$}{$x\,{\geq}\,y \wedge y\,{>}\,0 \wedge x\,{<}\,2$}};
\node[node, below of = 2a, xshift=-45mm] (3c)  
  {\pcnode{$\m b_2$}{$\psi_2'$}{$x\,{\geq}\,y \wedge y\,{>}\,0$}};
\node[node, below of = 3a] (4a)
  {\pcnode{$\m b_2$}{$\top$}{$x\,{\geq}\,y \wedge 0\,{<}\,y\,{<}\,2$}};
\node[node, below of = 3b] (4b)
  {\pcnode{$\m b_2$}{$\top$}{$x\,{\geq}\,y \wedge y\,{>}\,0 \wedge x_0\,{<}\,2$}};
\node[node, below of = 4a, accepting state] (5a)
  {\pcnode{$\m b_3$}{$\top$}{$x\eqn y \wedge 0\,{<}\,y\,{<}\,2$}};
\node[node, below of = 4b, accepting state] (5b)
  {\pcnode{$\m b_3$}{$\top$}{$x\eqn y \wedge y\,{>}\,0 \wedge x_0\,{<}\,2$}};
\draw[edge] ($(0) + (0,.4)$) -- (0);
\draw[edge] (0) -- (1a);
\draw[edge] (0) -- node[action, above, very near end, xshift=3mm, yshift=-1mm] {$\Conf_{x<2}$} (1b);
\draw[edge] (1a) -- (2a);
\draw[edge] (1b) -- (2b);
\draw[edge,dashed] (2a) -- ($(2a.west) + (-1,0)$);
\draw[edge,dashed] (3c) -- ($(3c) + (0,-1)$);
\draw[edge] (1a) -- node[action, above, xshift=-2mm] {$\Conf_{x<2}$} (2b);
\draw[edge] (2a) --node[action, above, xshift=1mm, near start] {$\Conf_{x<2}$} (3b);
\draw[edge] (2a) --node[action, left, xshift=-2mm] {$\Conf_{x<2}$} (3a);
\draw[edge] (2b) -- (3b);
\draw[edge, bend right=10] (2a.west) to (3c);
\draw[edge, bend left=38] (2b.east) to (4b.east);
\draw[->] (3c) to[loop left, min distance=5mm, out=177, in=183, looseness=8] (3c);
\draw[edge] (3c) --node[action, above, yshift=0mm] {$\Conf_{x<2}$} (3a);
\draw[edge] (3c.-5) --node[action, left, near end, yshift=0mm] {$\Conf_{x<2}$} (5a.175);
\draw[edge] (3a) -- (4a);
\draw[edge] (4a) -- (5a);
\draw[edge] (4b) -- (5b);
\draw[->] (4b) to[loop left, min distance=4mm, out=177, in=183, looseness=4] (4b);
\draw[->] (4a) to[loop left, min distance=4mm, out=3, in=-3, looseness=4] (4a);
\node[right of = 3b, scale=.7, red!80!black, xshift=15mm] {$\phi_1$};
\node[right of = 5a, scale=.7, red!80!black, xshift=9mm] {$\phi_2$};
\node[right of = 5b, scale=.7, red!80!black, xshift=12mm] {$\phi_3$};
\end{tikzpicture}
\end{center}
For $\vec U = \langle \hat x, \hat y\rangle$, and the formulas $\phi_1$, $\phi_2$, and $\phi_3$
in final nodes, we compute\\[.5ex]
\centerline{
$\begin{array}{rll}
\exists \vec U.\: \phi_1(\vec V, \vec U) &= \exists \hat x\,\hat y.\ 
 \hat x \eqn x\eqn \hat y \wedge \hat y > 0 \wedge \hat x < 2 &\equiv x > 0 \wedge x < 2 \\
\exists \vec U.\: \phi_2(\vec V, \vec U) &= \exists \hat x\,\hat y.\ 
 \hat x \eqn \hat y \wedge 0\,{<}\,\hat y\,{<}\, 2) &\equiv \top \\
\exists \vec U.\: \phi_3(\vec V, \vec U) &= \exists \hat x\,\hat y.\ 
 \hat x\eqn \hat y \wedge \hat y\,{>}\,0 \wedge x\,{<}\,2 &\equiv x < 2
\end{array}$}\\
so that $K_3=\chP(\psi_1)$ sets
$K_3(\m b_1) = \bigvee_{i=1}^3 \exists \vec U.\: \phi_i(\vec V, \vec U) \equiv \top$.
For state $\m b_3$, we get the following simple automaton:
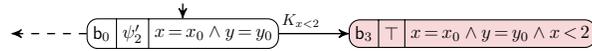
\begin{center}
\begin{tikzpicture}[node distance=10mm,>=stealth']
\tikzstyle{node} = [draw,rectangle split, rectangle split parts=3,rectangle split horizontal, rectangle split draw splits=true, inner sep=3pt, scale=.7, rounded corners]
\tikzstyle{goto} = [->]
\tikzstyle{action}=[scale=.6, black]
\tikzstyle{constr}=[scale=.5, black]
\tikzstyle{accepting state} = [fill=red!80!black!15]
\node[node] (0)  {\pcnode{$\m b_0$}{$\psi_2'$}{$x\eqn x_0 \wedge y \eqn y_0$}};
\node[node, right of = 0, xshift=45mm, accepting state] (1a)
  {\pcnode{$\m b_3$}{$\top$}{$x\eqn x_0 \wedge y \eqn y_0 \wedge x\,{<}\,2$}};
\draw[edge] ($(0) + (0,.4)$) -- (0);
\draw[edge] (0) -- node[action, above, xshift=1mm, near start] {$\Conf_{x<2}$} (1a);
\draw[edge,dashed] (0) -- ($(0.west) + (-1,0)$);
\end{tikzpicture}
\end{center}
For the formula $\phi$ in the final state we have
$\exists \vec U.\: \phi_3(\vec V, \vec U) = \exists \hat x\,\hat y.\ 
 \hat x\eqn x \wedge \hat y \eqn y \wedge x\,{<}\,2 \equiv x < 2$ so that
$K_3(\m b_3) = x < 2$.
\end{example}

The next example illustrates our approach on some simple properties
to illustrate how the branching time requirements are reflected.

\begin{example}
Let $\BB$ be the following simple DDSA:
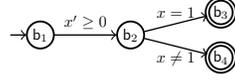
\begin{center}
\begin{tikzpicture}[node distance=20mm]
 \node[state] (0) {$\m b_1$};
\draw[edge] ($(0) + (-.4,0)$) -- (0);
 \node[state, right of=0] (1) {$\m b_2$};
 \node[state, right of=1, yshift=5mm, double] (2) {$\m b_3$};
 \node[state, right of=1, yshift=-5mm, double] (3) {$\m b_4$};
 \draw[edge] (0) -- node[action, above] {$x' \geq 0$} (1);
 \draw[edge] (1) -- node[action, above] {$x = 1$} (2);
 \draw[edge] (1) -- node[action, below] {$x \neq 1$} (3);
\end{tikzpicture}
\end{center}
\begin{itemize}
\item Consider $\psi_1 = \E\X ((x = 1) \wedge \E\X (x = 2))$.
We first evaluate $\E\X (x = 2)$ on all states.
The NFA for the formula $\psi_0 = \X K_{x = 2}$ is as follows, for $q_1 = \psi_1$ and
$q_2 = K_{x = 2}$:
\begin{center}
\begin{tikzpicture}[node distance=25mm]
\draw[edge] ($(0) + (-.4,0)$) -- (0);
 \node[state, minimum width=6mm] (0) {$q_1$};
 \node[state, right of=0, minimum width=6mm] (1) {$q_2$};
 \node[state, right of=1, double, minimum width=6mm] (2) {$\top$};
 \draw[edge] (0) -- (1);
 \draw[edge] (1) -- node[action, above] {$\{K_{x = 2}\}$} (2);
\draw[->] (2) to[loop above, looseness=6]  (2);
\end{tikzpicture}
\end{center}
This leads to the product automata shown next:
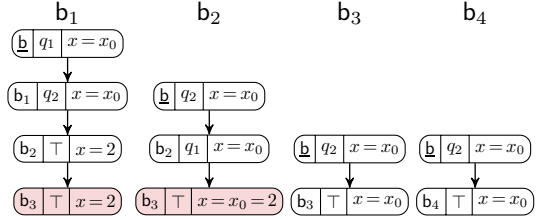
\begin{center}
\begin{tabular}{cccc}
$\m b_1$ & $\m b_2$ & $\m b_3$ & $\m b_4$ \\
\begin{tikzpicture}[node distance=10mm,>=stealth']
\tikzstyle{node} = [draw,rectangle split, rectangle split parts=3,rectangle split horizontal, rectangle split draw splits=true, inner sep=3pt, scale=.7, rounded corners]
\tikzstyle{goto} = [->]
\tikzstyle{action}=[scale=.6, black]
\tikzstyle{constr}=[scale=.5, black]
\tikzstyle{accepting state} = [fill=red!80!black!15]
\node[node] (0)  {\pcnode{$\bdummy[\m b]$}{$q_1$}{$x\,{=}\,x_0$}};
\node[node, below of = 0] (1)  {\pcnode{$\m b_1$}{$q_2$}{$x\,{=}\,x_0$}};
\node[node, below of = 1] (2)  {\pcnode{$\m b_2$}{$\top$}{$x\,{=}\,2$}};
\node[node, below of = 2, accepting state] (3)  {\pcnode{$\m b_3$}{$\top$}{$x\,{=}\,2$}};
\draw[edge] (0) -- (1);
\draw[edge] (1) -- (2);
\draw[edge] (2) -- (3);
\end{tikzpicture}
&
\begin{tikzpicture}[node distance=10mm,>=stealth']
\tikzstyle{node} = [draw,rectangle split, rectangle split parts=3,rectangle split horizontal, rectangle split draw splits=true, inner sep=3pt, scale=.7, rounded corners]
\tikzstyle{goto} = [->]
\tikzstyle{action}=[scale=.6, black]
\tikzstyle{constr}=[scale=.5, black]
\tikzstyle{accepting state} = [fill=red!80!black!15]
\node[node, below of = 0] (1)  {\pcnode{$\bdummy[\m b]$}{$q_2$}{$x\,{=}\,x_0$}};
\node[node, below of = 1] (2)  {\pcnode{$\m b_2$}{$q_1$}{$x\,{=}\,x_0$}};
\node[node, below of = 2, accepting state] (3)  {\pcnode{$\m b_3$}{$\top$}{$x\,{=}\,x_0\,{=}\,2$}};
\draw[edge] (1) -- (2);
\draw[edge] (2) -- (3);
\end{tikzpicture}
&
\begin{tikzpicture}[node distance=10mm,>=stealth']
\tikzstyle{node} = [draw,rectangle split, rectangle split parts=3,rectangle split horizontal, rectangle split draw splits=true, inner sep=3pt, scale=.7, rounded corners]
\tikzstyle{goto} = [->]
\tikzstyle{action}=[scale=.6, black]
\tikzstyle{constr}=[scale=.5, black]
\tikzstyle{accepting state} = [fill=red!80!black!15]
\node[node, below of = 0] (1)  {\pcnode{$\bdummy[\m b]$}{$q_2$}{$x\,{=}\,x_0$}};
\node[node, below of = 1] (2)  {\pcnode{$\m b_3$}{$\top$}{$x\,{=}\,x_0$}};
\draw[edge] (1) -- (2);
\end{tikzpicture}
&
\begin{tikzpicture}[node distance=10mm,>=stealth']
\tikzstyle{node} = [draw,rectangle split, rectangle split parts=3,rectangle split horizontal, rectangle split draw splits=true, inner sep=3pt, scale=.7, rounded corners]
\tikzstyle{goto} = [->]
\tikzstyle{action}=[scale=.6, black]
\tikzstyle{constr}=[scale=.5, black]
\tikzstyle{accepting state} = [fill=red!80!black!15]
\node[node, below of = 0] (1)  {\pcnode{$\bdummy[\m b]$}{$q_2$}{$x\,{=}\,x_0$}};
\node[node, below of = 1] (2)  {\pcnode{$\m b_4$}{$\top$}{$x\,{=}\,x_0$}};
\draw[edge] (1) -- (2);
\end{tikzpicture}
\end{tabular}
\end{center}
For $\m b_1$,  line 6 of $\chP$ thus yields $K'(\m b_1) = (\exists x. x=2)(\vec V) = \top$, and $K'(\m b_2) = (\exists x. x=2 \wedge x_0 = 2)(\vec V) = (x = 2)$.
Overall, the evaluation of $\E\X (x = 2)$ thus yields $K'$ such that 
$K' = \{\m b_1 \mapsto \top,\ \m b_2 \mapsto (x = 2), \m b_3 \mapsto \bot,\ \m b_4 \mapsto \bot\}$.
We hence construct the NFA for the formula $\psi_1' = \X (K_{x = 1} \wedge K')$, 
which looks as follows, for $q_1 =\psi_1'$ and $q_2 = K_{x = 1} \wedge K'$:
\begin{center}
\begin{tikzpicture}[node distance=25mm]
 \node[state, minimum width=6mm] (0) {$q_1$};
 \node[state, right of=0, minimum width=6mm] (1) {$q_2$};
 \node[state, right of=1, double, minimum width=6mm] (2) {$\top$};
 \draw[edge] (0) -- (1);
 \draw[edge] (1) -- node[action, above] {$\{K_{x = 1},  K'\}$} (2);
\draw[edge] ($(0) + (-.4,0)$) -- (0);
\end{tikzpicture}
\end{center}
We focus now on the evaluation of $\psi_1$ in state $\m b_1$.
The product construction for $\BB$, $\m b_1$, and $\psi_1'$ is started as follows:\\
\centerline{
\begin{tikzpicture}[node distance=30mm,>=stealth']
\tikzstyle{node} = [draw,rectangle split, rectangle split parts=3,rectangle split horizontal, rectangle split draw splits=true, inner sep=3pt, scale=.7, rounded corners]
\tikzstyle{goto} = [->]
\tikzstyle{action}=[scale=.6, black]
\tikzstyle{constr}=[scale=.5, black]
\node[node] (0)  {\pcnode{$\m b_0$}{$q_1$}{$x\,{=}\,x_0$}};
\node[node, right of = 0] (1)  {\pcnode{$\m b_1$}{$q_2$}{$x\,{=}\,x_0$}};
\draw[edge] (0) -- (1);
\end{tikzpicture}
}\\
However, at this point the next product transition would combine 
$\m b_1 \goto{x' \geq 0} \m b_2$ with
$\m q_2 \goto{\{K_{x = 1},  K'\}} \m \top$:
The labels $K_{x = 1}$ and $K_1$ evaluated at the destination state $\m b_2$,
yield the constraints $x = 1$ and $x=2$, but since their conjunction is unsatisfiable,
the product construction stops at this point. Hence there are no final states, so that 
the resulting configuration map $K:=\chP(\psi_1)$ sets $K(\m b_1) = \bot$, as expected.
\item
Consider $\psi_2 = (\E\X (x = 1)) \wedge (\E\X (x = 2))$.
To recursively process the formula, we first evaluate $\E\X (x = 1)$ on all states as above,
which yields
$K_1 = \{\m b_1 \mapsto \top,\ \m b_2 \mapsto (x = 1), \m b_3 \mapsto \bot,\ \m b_4 \mapsto \bot\}$, and similarly for $\E\X (x = 1)$ we obtain 
$K_2 = \{\m b_1 \mapsto \top,\ \m b_2 \mapsto (x = 2), \m b_3 \mapsto \bot,\ \m b_4 \mapsto \bot\}$.
When evaluating $\psi_2$, we hence return 
$K_1 \wedge K_2 = \{\m b_1 \mapsto \top,\ \m b_2 \mapsto \bot, \m b_3 \mapsto \bot,\ \m b_4 \mapsto \bot\}$.
\item 
Consider $\psi_3 = \E \X(\E\X (x = 1) \wedge \E\X (x = 2))$.
As above, we first evaluate $K_1$ and $K_2$ as above.
We then construct the NFA for the formula $\psi_3' = \X (K_1 \wedge K_2)$, 
which looks as follows, for $q_1 =\psi_3'$ and $q_2 = K_1 \wedge K_2$:\\
\centerline{
\begin{tikzpicture}[node distance=25mm]
 \node[state, minimum width=6mm] (0) {$q_1$};
 \node[state, right of=0, minimum width=6mm] (1) {$q_2$};
 \node[state, right of=1, double, minimum width=6mm] (2) {$\top$};
 \draw[edge] (0) -- (1);
 \draw[edge] (1) -- node[action, above] {$\{K_1,  K_2\}$} (2);
\draw[edge] ($(0) + (-.4,0)$) -- (0);
\end{tikzpicture}
} \\
To evaluate $\psi_3$ in state $\m b_1$
the product construction for $\BB$, $\m b_1$, and $\psi_3'$ is again started as follows:\\
\centerline{
\begin{tikzpicture}[node distance=30mm,>=stealth']
\tikzstyle{node} = [draw,rectangle split, rectangle split parts=3,rectangle split horizontal, rectangle split draw splits=true, inner sep=3pt, scale=.7, rounded corners]
\tikzstyle{goto} = [->]
\tikzstyle{action}=[scale=.6, black]
\tikzstyle{constr}=[scale=.5, black]
\node[node] (0)  {\pcnode{$\m b_0$}{$q_1$}{$x\,{=}\,x_0$}};
\node[node, right of = 0] (1)  {\pcnode{$\m b_1$}{$q_2$}{$x\,{=}\,x_0$}};
\draw[edge] (0) -- (1);
\end{tikzpicture}
} \\
However, at this point the next product transition would combine 
$\m b_1 \goto{x' \geq 0} \m b_2$ with
$\m q_2 \goto{\{K_1,  K_2\}} \m \top$:
The labels $K_1$ and $K_2$ evaluated at the desination state $\m b_2$
yield the constraints $x = 1$ and $x=2$, and since their conjunction is unsatisfiable,
the product construction stops at this point, without producing a final state. Hence 
the resulting configuration map $K$ sets $K(\m b_1) = \bot$.
\end{itemize}
\end{example}

\section{NFA Construction}
\label{app:nfa}

For the following construction, we assume that $\psi \in \LTLconf$ is in negation normal form.
To this end we need to extend the grammar for $\LTLconf$ formulas to allow disjunction $\psi_1 \vee \psi_2$ and a weak next operator $\Xw \psi$. The semantics \defref{ltl:semantics} is extended as
$\rho,i \models \psi_1 \vee \psi_2$ iff 
$\rho,i \models \psi_1$ or $\rho,i \models \psi_2$, and
$\rho,i \models \Xw \psi$ iff $i = n$ or $\rho,i+1 \models \psi$.
Then a formula $\neg \X \psi$ can be written as $\Xw \neg \psi$, so that we can assume
$\psi$ to be in negation normal form.
We can assume that the only base case is $K\in \Confs(\Phi)$
because for every $K$ also $\neg K$ is in $\Confs(\Phi)$.

We build an NFA $\NFApsi = (Q, \Sigma, \varrho, q_0, Q_F)$, where
\begin{inparaenum}[\it (i)]
\item the set $Q$ of states is a set of quoted $\LTLconf$ formulas together with $\{\inquotes{\top}, \inquotes{\bot}\}$; 
\item $\Sigma\,{=}\, 2^{\Confs(\Phi)}$ is the alphabet;
\item $\varrho \subseteq Q \times \Sigma \times Q$ is the transition relation; 
\item $q_0 \in Q$ is the initial state;
\item $Q_F\subseteq Q$ is the set of final states. 
\end{inparaenum}

Following~\cite{GMM14}, we define $\varrho$ using an auxiliary function $\delta$
and a new proposition $\last$ that marks the last element of the trace.
The input of $\delta$ is a 
(quoted) formula $\psi \in \LTLconf$, and its output a set of tuples
$(\inquotes{\psi'},\varsigma)$ where $\psi'$ has the same type as $\psi$ and 
$\varsigma \in 2^{S\cup \{\last, \neg \last\}}$.
For two sets of such tuples $R_1$, $R_2$, and $\odot$ either $\wedge$ or $\vee$, let
$R_1 \odot R_2 = \{ (\inquotes{\psi_1 \odot \psi_2}, \varsigma_1 \cup \varsigma_2) \mid (\inquotes{\psi_1}, \varsigma_1) \inn R_1, (\inquotes{\psi_2}, \varsigma_2) \inn R_2 \}$, where we simplify $\psi_1 \odot \psi_2$ if possible. 
The function $\delta$ is as follows:

\noindent
\begin{tabular}{@{~}l@{~}l}
$\delta(\inquotes{\top})$ &= 
  $\{(\inquotes{\top},\emptyset)\}
  \text{ and }\delta(\inquotes{\bot}) = \{(\inquotes{\bot},\emptyset)\}$\\
$\delta(\inquotes{K})$ &= 
  $\{(\inquotes{\top},\{K\}),(\inquotes{\bot},\emptyset)\} 
  \text{ if $K \in \Confs(\Phi)$}$\\
$\delta(\inquotes{\psi_1 \vee \psi_2})$ &=
  $\delta(\inquotes{\psi_1}) \vee \delta(\inquotes{\psi_2})$ \\
$\delta(\inquotes{\psi_1 \wedge \psi_2})$ &= 
  $\delta(\inquotes{\psi_1}) \wedge \delta(\inquotes{\psi_2})$ \\
$\delta(\inquotes{\X \psi})$ &= 
  $\{(\inquotes{\psi},\{\neg \last\}), (\inquotes{\bot}, \{\last\})\}$ \\
$\delta(\inquotes{\Xw \psi})$ &= 
  $\{(\inquotes{\psi},\{\neg \last\}), (\inquotes{\top}, \{\last\})\}$ \\
$\delta(\inquotes{\G \psi})$ &=
  $\delta(\inquotes{\psi}) \wedge (\delta(\inquotes{\X\G \psi}) \vee \delta_\lambda)$ \\
$\delta(\inquotes{\psi_1 \U \psi_2})$ &=
  $\delta(\inquotes{\psi_2}) \vee (\delta(\inquotes{\psi_1})
  \wedge \delta(\inquotes{\X(\psi_1 \U \psi_2)}))$
\end{tabular}

\noindent
where $\delta_\lambda$ abbreviates 
$\smash{\{(\inquotes{\top},\{ \lambda\}),(\inquotes{\bot},\{\neg\lambda\})\}}$.
While the symbol $\last$ is needed for the construction, we can omit it
from the NFA, and define $\NFApsi$ as follows:
\begin{definition}
\label{def:NFA}
For a formula $\psi \inn \LTLconf$, let the NFA 
$\NFApsi\,{=}\,(Q, \Sigma, \varrho, q_0, \{q_f, q_e\})$
be given by
$q_0\,{=}\,\inquotes{\psi}$,
$q_f\,{=}\,\inquotes{\top}$ and $q_e$ is an additional final state,
and
$Q$, $\varrho$ are the smallest sets such that $q_0, q_f, q_e \in Q$ and whenever 
$q\in Q\setminus\{q_e\}$ and $(q', \varsigma)\in \delta(q)$
such that $\{\last,\neg \last\} \not\subseteq \varsigma$
then $q'\in Q$ and 
\begin{compactenum}[(i)]
\item if $\last \not\in \varsigma$ then
$(q, \varsigma \setminus\{\last, \neg \last\}, q') \in \varrho$, and
\item
if $\last \in \varsigma$ and $q' = \inquotes{\top}$ then
$(q, \varsigma \setminus\{\last, \neg \last\}, q_e) \in \varrho$.
\end{compactenum}
\end{definition}

In order to express correctness, more precise consistency notions are required.
Let $\Sigma' =  2^{\Confs(\Phi) \cup\{\last, \neg \last\}}$.
Then,
$\varsigma \inn \Sigma$ is 
\emph{consistent with step $i$} of a run
\begin{equation}
\label{eq:therun}
\rho\colon
(b_0, \alpha_0) \goto{a_1}
(b_1, \alpha_1) \goto{a_2} \dots 
\goto{} (b_n, \alpha_n)
\end{equation}
if $\alpha_i \models \varsigma(b_i)$. Moreover,
$\varsigma \inn \Sigma'$ is \emph{$\lambda$-consistent with step $i$} of $\rho$
if $\varsigma$ is consistent with step $i$ of $\rho$, 
if $i<n$ then $\last \not\in \varsigma$, and 
if $i=n$ then $\neg \last \not\in \varsigma$.
By definition, a word $\varsigma_0 \varsigma_1 \cdots \varsigma_n\in \Sigma^*$ is consistent
with a run $\rho$ if $\varsigma_i$ is consistent with step $i$ of $\rho$ for all 
$i$, $0\,{\leq}\,i\,{\leq}\,n$.

We first note that the function $\delta$ is total in the sense that 
for every assignment $\alpha$ and run $\rho$,
the returned set has an entry that is $\lambda$-consistent with $\alpha$ and $\rho$.

\begin{lemma}
\label{lem:delta total}
For every run $\rho$ of the form \eqref{therun}, every $i$, $0\leq i \leq n$,
and $\psi \in \LTLconf$,
there is some $(\inquotes{\psi'}, \varsigma) \in \delta(\inquotes{\psi})$ 
such that 
$\varsigma$ is $\lambda$-consistent with step $i$ of $\rho$.
\end{lemma}
\begin{proof}
By induction on the structure of $\psi$ using the definition of $\delta$.
The claim is easy to check for every base
case of the definition of $\delta$, and in all other cases it follows from the induction hypothesis.
\qed
\end{proof}

\smallskip
We next show a crucial feature of the $\delta$ function,
namely that it preserves and reflects the property of a run satisfying a formula.
Both directions are proven by tedious but straightforward induction proofs on the formula
structure.

\begin{lemma}
\label{lem:delta}
Let $\psi \in \LTLconf \cup \{\top,\bot\}$,
$\rho$ a run of the form \eqref{therun}, and $0\,{\leq}\,i\,{\leq}\,n$.
Then
$\rho,i \modelsLTL \psi$ holds if and only if
there is some $(\inquotes{\psi'}, \varsigma)\in \delta(\inquotes{\psi})$ such that\\
\noindent
\begin{tabular}{@{\ }r@{\ }p{8cm}}
$(a)$ & $\varsigma$ is $\lambda$-consistent with step $i$ of $\rho$, \\
$(b)$ &either $i<n$ and $\rho,i{+}1 \modelsLTL \psi'$, or $i\,{=}\,n$ and $\psi' = \top$.
\end{tabular}
\end{lemma}
\begin{proof}
($\Longrightarrow$)
We first note that if $\psi' = \top$ then $(b)$ holds for both $i<n$ and $i=n$ ($\star$).
The proof is by induction on $\psi$.
\begin{compactitem}
\item
If $\psi = \top$, we can choose $(\inquotes{\psi'},\varsigma) = (\top, \emptyset)$.
Then, $\emptyset$ is $\lambda$-consistent with any step, and (b) follows from $(\star)$.
\item
If $\rho,i \modelsLTL K$ for some $K \in \Confs(\Phi)$, we may take
$(\inquotes{\top}, \{K\})\in \delta(\inquotes{K})$.
As $\rho,i \modelsLTL K$, $\alpha_i$ satisfies 
$K(b_i)$, so consistency holds and we use ($\star$) for $(b)$.
\item 
If $\rho,i \modelsLTL \X \psi$ then
$i\,{<}\,n$ and $\rho,i{+}1 \modelsLTL \psi$.
For
$(\inquotes{\psi}, \{\neg \last\}) \in \delta(\inquotes{\X \psi})$,
part $(a)$ holds since 
$\neg \last \in \varsigma$ and $i\,{<}\,n$, and
$(b)$ because of $\rho,i{+}1 \modelsLTL \psi$.
\item 
Suppose $\rho,i \modelsLTL \Xw \psi$. If $i=n$ then (a) is by definition, and (b) by $(\star)$.
If $i\,{<}\,n$ then $\rho,i{+}1 \modelsLTL \psi$.
For
$(\inquotes{\psi}, \{\neg \last\}) \in \delta(\inquotes{\Xw \psi})$,
part $(a)$ holds since 
$\neg \last \in \varsigma$ and $i\,{<}\,n$, and
$(b)$ because of $\rho,i{+}1 \modelsLTL \psi$.
\item 
Suppose $\psi = \psi_1\wedge \psi_2$.
By assumption $\rho,i \modelsLTL \psi_1\wedge \psi_2$, and hence 
$\rho,i \modelsLTL \psi_1$ and $\rho,i \modelsLTL \psi_2$.
By the induction hypothesis, there are
$(\inquotes{\psi_1'}, \varsigma_1) \in \delta (\inquotes{\psi_1})$ and
$(\inquotes{\psi_2'}, \varsigma_2) \in \delta (\inquotes{\psi_2})$ such that 
for both $k\in \{1,2\}$,
$(a')$ $\varsigma_k$ is consistent with step $i$ of $\rho$ and
$(b')$ either $\rho,i{+}1 \modelsLTL \psi_k'$, or $i\,{=}\,n$ and $\psi_k' = \top$.
By definition of $\delta$,  we can choose 
$(\inquotes{\psi_1' \wedge \psi_2'}, \varsigma_1 \cup \varsigma_2) \in 
\delta(\inquotes{\psi_1 \wedge \psi_2})$.
Then
$(a)$ follows from $(a')$ and $\varsigma = \varsigma_1 \cup \varsigma_2$, and
$(b)$ if $i=n$ then $(b')$ implies $\psi' = \top$, and otherwise
$\rho,i{+}1 \modelsLTL \psi_1' \wedge \psi_2'$.
\item 
Suppose $\psi = \psi_1\vee \psi_2$.
By assumption $\rho,i \modelsLTL \psi_1\vee \psi_2$, and hence 
$\rho,i \modelsLTL \psi_1$ or $\rho,i \modelsLTL \psi_2$.
We assume the former.
By the induction hypothesis, there is some
$(\inquotes{\psi_1'}, \varsigma_1) \in \delta (\inquotes{\psi_1})$ such that 
$(a')$ $\varsigma_1$ is consistent with step $i$ of $\rho$, and
$(b')$ $\rho,i{+}1 \modelsLTL \psi_1'$, or $i\,{=}\,n$ and $\psi_1' = \top$.
As $\delta$ is total (\lemref{delta total}), there must be some 
$(\inquotes{\psi_2'}, \varsigma_2) \in \delta (\inquotes{\psi_2})$ such that
$\varsigma_2$ is $\lambda$-consistent with step $i$ of $\rho$.
By definition of $\delta$,  we can choose 
$(\inquotes{\psi'},\varsigma)$ as 
$(\inquotes{\psi_1' \vee \psi_2'}, \varsigma_1 \cup \varsigma_2) \in 
\delta(\inquotes{\psi_1 \vee \psi_2})$.
Then
$(a)$ follows from $(a')$ and $\varsigma_2$ being $\lambda$-consistent with step $i$ of $\rho$, and
$(b)$ if $i\eqn n$ then $(b')$ implies $\psi_1'\eqn \top$, hence $\psi' \eqn \top$; otherwise $\rho,i{+}1 \modelsLTL \psi_1' \vee \psi_2'$.
\item
Suppose $\rho,i \modelsLTL \G \psi$, so $\rho,i \modelsLTL \psi$ and either 
(1) $i=n$, or 
(2) $\rho,i{+}1 \modelsLTL \G \psi$.
We have 
$\delta(\inquotes{\G \psi}) = \delta(\inquotes{\psi}) \wedge (\delta(\inquotes{\langle\cdot\rangle \G \psi}) \vee \delta_\last) = 
(\delta(\inquotes{\psi}) \wedge \delta(\inquotes{\langle\cdot\rangle \G \psi}))
\vee (\delta(\inquotes{\psi}) \wedge \delta_\last)$.
In either case, by the induction hypothesis there is some
$(\inquotes{\psi'}, \varsigma') \in \delta (\inquotes{\psi})$ such that
$(a')$ $\varsigma'$ is $\last$-consistent with step $i$ of $\rho$, and
$(b')$ $\rho,i{+}1 \modelsLTL \psi'$, or $i\,{=}\,n$ and $\psi' = \top$.

(1) Let
$(\inquotes{\psi_1},\varsigma_1)$ be
$(\inquotes{\psi' \wedge \top}, \varsigma' \cup \{ \last \}) \in 
\delta(\inquotes{\psi}) \wedge \delta_\last$. We have
$(a_1)$ $\varsigma_1$ is $\last$-consistent with step $i$ of $\rho$ because of $(a')$ and $i\eqn n$, and
$(b_1)$ $\rho,i{+}1 \modelsLTL \psi_1 = \top$ by $(b')$.

(2) Let 
$(\inquotes{\psi_2},\varsigma_2)$ be
$(\inquotes{\psi' \wedge \G \psi}, \varsigma' \cup \{ \neg \last \}) \in 
\delta(\inquotes{\psi}) \wedge \delta(\inquotes{\langle\cdot\rangle \G \psi})$.
Then
$(a_2)$ $\varsigma_2$ is $\last$-consistent with step $i$ of $\rho$ by $(a')$ and $i<n$, and
$(b_2)$ $\rho,i{+}1 \modelsLTL \psi_2 = \psi' \wedge \G \psi$, using $(b')$ and $\rho,i{+}1 \modelsLTL \G \psi$.
Thus the two cases can be combined as in the case for disjunction,
using $(a_1)$, $(b_1)$ and $(a_2)$, $(b_2)$.
\item The case for the $\U$ operator is similar.
\end{compactitem}
\noindent
($\Longleftarrow$)
Note that the assumptions exclude $\psi' = \bot$.
We apply induction on $\psi$, and use the definition of $\delta$ for each case.
\begin{compactitem}
\item 
If $\psi = \top$ then $(\inquotes{\psi'}, \varsigma)\in \delta(\inquotes{\psi})$ implies $\psi' = \top$, and 
$\rho, i \modelsLTL \top$ holds.
\item
If $\psi\,{=}\,K\in \Confs(\Phi)$, we must have
$\psi' = \top$ and $\varsigma =\{K\}$.
As $\alpha_{i}\models K(b_i)$ by $\last$-consistency, $\rho,i \modelsLTL K$.
\item 
Let $\psi = \X \chi$.
As $\psi'=\bot$ or $\rho, i{+}1 \modelsLTL \psi'$, 
by definition of $\delta$ the only possibility is
$\psi' = \inquotes{\chi}$ and $\varsigma = \{\neg \last\}$.
As $\varsigma$ is consistent with step $i$ of $\rho$ and $\neg \last \in \varsigma$,
we must have $i\,{<}\,n$, so $\rho, i{+}1 \modelsLTL \psi'$
and hence $\rho, i \modelsLTL \psi$ by \defref{ltl:semantics}.
\item 
Suppose $\psi = \Xw \chi$.
If $\psi'=\top$ then $i=n$ and $\rho, i \modelsLTL \psi$ holds by definition.
Otherwise, we can reason as in the case above.
\item
If $\psi = \psi_1 \wedge \psi_2$ then 
by $(\inquotes{\psi'}, \varsigma) \in \delta(\inquotes{\psi})$
and the definition of $\delta$ there are $\psi_1'$ and $\psi_2'$ such that
$(\inquotes{\psi_1'}, \varsigma_1')\in \delta(\inquotes{\psi_1})$ and $(\inquotes{\psi_2'}, \varsigma_2')\in \delta(\inquotes{\psi_2})$, and
$\psi' = \psi_1' \wedge \psi_2'$ and $\varsigma = \varsigma_1' \cup \varsigma_2'$.
Therefore, either $i=n$ and $\psi' = \psi_1' = \psi_2' = \top$, or
$i<n$ and $\rho, i{+}1 \modelsLTL \psi'$, which implies
$\rho, i{+}1 \modelsLTL \psi_1'$ and $\rho, i{+}1 \modelsLTL \psi_2'$.
In either case, 
$\rho, i \modelsLTL \psi_1$ and $\rho, i \modelsLTL \psi_2$ 
hold by the induction hypothesis,
so $\rho, i \modelsLTL \psi_1 \wedge \psi_2$.
\item
Similarly,
if $\psi = \psi_1 \vee \psi_2$ then there are $\psi_1'$ and $\psi_2'$ such that
$(\inquotes{\psi_1'}, \varsigma_1')\in \delta(\inquotes{\psi_1})$ and $(\inquotes{\psi_2'}, \varsigma_2')\in \delta(\inquotes{\psi_2})$,
$\psi' = \psi_1' \vee \psi_2'$ and $\varsigma = \varsigma_1' \cup \varsigma_2'$.
If $i=n$ and $\psi' = \top$, then $\psi_1' = \top $ or $\psi_2' = \top$.
If otherwise $i<n$ then $\rho, i{+}1 \modelsLTL \psi'$ implies
$\rho, i{+}1 \modelsLTL \psi_1'$ or $\rho, i{+}1 \modelsLTL \psi_2'$.
From the induction hypothesis we obtain in either case
$\rho, i \modelsLTL \psi_1$ or $\rho, i \modelsLTL \psi_2$,
so $\rho, i \modelsLTL \psi_1 \vee \psi_2$.
\item
If $\psi = \G \chi$ then we can distinguish two cases:

(1) There are $\psi_1$ and $\psi_2$ such that
$(\inquotes{\psi_1'}, \varsigma_1)\in \delta(\inquotes{\chi})$, $(\inquotes{\psi_2'}, \varsigma_2)\in \delta_{\last}$, 
$\psi' = \psi_1' \wedge \psi_2'$ and $\varsigma = \varsigma_1 \cup \varsigma_2$.
As $(\inquotes{\psi_2'}, \varsigma_2)\in \delta_{\last}$, we must have $\psi_2' = \top$ and
$\varsigma_2 = \{\last\}$ (otherwise, we would have $\psi'=\bot$).
By consistency, $\last \in \varsigma$ implies $i = n$, so $\psi' = \top$
by assumption and therefore we must have $\psi_1' = \top$.
From the induction hypothesis and
$(\inquotes{\psi_1'}, \varsigma_1)\in \delta(\inquotes{\chi})$ 
we conclude $\rho, i \modelsLTL \chi$, so by \defref{ltl:semantics}
$\rho, i \modelsLTL \G \chi$.

(2)
There are $\psi_1'$ and $\psi_2'$ such that
$(\inquotes{\psi_1'}, \varsigma_1)\in \delta(\inquotes{\chi})$, 
$(\inquotes{\psi_2'}, \varsigma_2)\in \delta(\inquotes{\langle\cdot\rangle\G \chi})$,
$\psi' = \psi_1' \wedge \psi_2'$ and $\varsigma = \varsigma_1 \cup \varsigma_2$.
We have $\neg\last \in \varsigma_2$, so by consistency $i<n$.
As $\rho, i{+}1 \modelsLTL \psi' = \psi_1' \wedge \psi_2'$, by \defref{ltl:semantics}  
$\rho, i{+}1 \modelsLTL \psi_1'$ and $\rho, i{+}1 \modelsLTL \psi_2'$.
By the induction hypothesis, 
$(\inquotes{\psi_1'}, \varsigma_1)\in \delta(\inquotes{\chi})$ and
$\rho, i{+}1 \modelsLTL \psi_1'$
imply $\rho, i \modelsLTL \chi$.
Moreover,
$(\inquotes{\psi_2'}, \varsigma_2)\in \delta(\inquotes{\langle\cdot\rangle\G \chi})$
and 
$\rho, i{+}1 \modelsLTL \psi_2'$ imply 
$\psi_2' = \G \chi$ by \defref{ltl:semantics}, so we have
$\rho, i \modelsLTL \langle\cdot\rangle\G \chi$.
Thus $\rho, i \modelsLTL \G \chi$.
\item
The case for $\U$ is similar. 
\qed
\end{compactitem}
\end{proof}

\noindent
Let a word $\varsigma_0 \varsigma_1 \cdots \varsigma_{n}\in \Sigma'^*$ be \emph{well-formed}
if $\last \not\in \varsigma_i$ for all $0\,{\leq}\,i\,{\leq}\,n$, and $\neg \last\,{\not\in}\,\varsigma_n$.

\begin{lemma}
\label{lem:deltastar}
A well-formed word $w\in \Sigma'^*$ that is consistent with a run $\rho$
satisfies $\inquotes{\top} \in \delta^*(\inquotes{\psi},w)$
iff
$\rho,0 \modelsLTL \psi$.
\end{lemma}
\begin{proof}
($\Longrightarrow$)
Let $w = \varsigma_0 \varsigma_1 \cdots \varsigma_{n}$ and
$\chi_0,\chi_1,\dots, \chi_{n+1}$
be the sequence of formulas witnessing
$\inquotes{\top} \in \delta^*(\inquotes{\psi},w)$,
so that $\chi_0 = \psi$, $\chi_{n+1} = \top$,
and $(\inquotes{\chi_{i{+}1}}, \varsigma_{i}) \in \delta(\inquotes{\chi_i})$
for all $i$, $0\leq i \leq n$.
As $w$ is well-formed and consistent with $\rho$, 
by definition $\varsigma_i$ is $\lambda$-consistent with 
$\rho$ at $i$ for all $i$, $0\leq i \leq n$.
In order to show that $\rho,0 \modelsLTL \psi$ holds, we verify that
$\rho, i \modelsLTL \chi_i$
for all $i$, $0\,{\leq}\,i\,{\leq}\,n$, by induction on $n-i$.
In the base case $i\,{=}\,n$.
We have $\chi_{n+1} = \top$ and 
$(\inquotes{\chi_{n{+}1}}, \varsigma_{n}) \in \delta(\inquotes{\chi_n})$,
and from \lemref{delta}\:($\Longrightarrow$) it follows that $\rho, n \modelsLTL \chi_n$.
If $i < n$, we assume by the induction hypothesis that $\rho, i{+}1 \modelsLTL \chi_{i{+}1}$.
We have
$(\inquotes{\chi_{i{+}1}}, \varsigma_{i}) \in \delta(\inquotes{\chi_i})$, so
$\rho, i \modelsLTL \chi_i$ follows again from \lemref{delta}\:($\Longrightarrow$), which concludes the induction step.
Finally, the claim follows for the case $i\,{=}\,0$ because 
$\chi_0 = \psi$.

($\Longleftarrow$)
Let $\rho$ be of the form \eqref{therun}.
We show that for all $i$, $0 \leq i \leq n$,
and every formula $\chi$,
if $\rho, i \modelsLTL \chi$ then 
there is a word $w_i = \varsigma_i \varsigma_{i+1} \cdots \varsigma_{n}$
of length $n-i+1$
such that
$\inquotes{\top} \in \delta^*(\inquotes{\chi},w_i)$,
and 
$\varsigma_j$ is $\lambda$-consistent with step $j$ of $\rho$ for all $j$, $i\leq j \leq n$.
The proof of is by induction on $n-i$.

In the base case where $i=n$, we assume that $\rho,n \modelsLTL \chi$. 
By \lemref{delta}\:($\Longleftarrow$) there is some $\varsigma_n$ such that
$(\inquotes{\top}, \varsigma_n)\in \delta(\inquotes{\chi})$,
and $\varsigma_n$ is $\lambda$-consistent with step $n$ of $\rho$.
For the induction step, assume $i\,{<}\,n$ and $\rho, i \modelsLTL \chi$.
By \lemref{delta}\:($\Longleftarrow$) there is some
$(\inquotes{\chi'}, \varsigma_i) \in \delta(\inquotes{\chi})$ such that
$\rho, i{+}1 \modelsLTL \chi'$, and moreover
$\varsigma_i$ is $\lambda$-consistent with $\rho$ at step $i$.
By the induction hypothesis, 
there is a word $w_{i+1} = \varsigma_{i+1} \cdots \varsigma_{n}$
such that
$\inquotes{\top} \in \delta^*(\inquotes{\chi'},w_{i+1})$,
and 
$\varsigma_j$ is $\lambda$-consistent with $\rho$ at instant $j$ for all $j$, $i\,{<}\,j\,{\leq}\, n$.
We can define $w_{i} = \varsigma_i\varsigma_{i+1} \cdots \varsigma_{n}$,
which satisfies
$\inquotes{\top} \in \delta^*(\inquotes{\chi},w_{i})$ and 
$\varsigma_j$ is $\lambda$-consistent with $\rho$ at $j$ for all $j$, $i\,{\leq}\,j\,{\leq}\, n$,
so the induction step works.

By assumption, $\rho,0 \modelsLTL \psi$  holds.
From the case $i = 0$ of the above statement, we obtain a word $w=w_0$
such that
$\inquotes{\top} \in \delta^*(\inquotes{\psi},w)$ and $w$ is
$\lambda$-consistent with all steps of $\rho$, i.e., $w$ is well-formed and consistent with $\rho$.
\qed
\end{proof}

We next show some simple properties that will be useful to show correctness of the 
automaton without $\last$.

\begin{lemma}
\label{lem:delta last}
Let $\psi \in \LTLconf$ and
$(\chi,\varsigma) \in \delta(\inquotes{\psi})$.
\begin{inparaenum}
\item[(1)] If $\last \inn \varsigma$ and $\neg \last \not\in\varsigma$
then $\chi \eqn \top$ or $\chi = \bot$.
\item[(2)] Suppose $\neg \last \inn \varsigma$, $\last \not\in\varsigma$,
and $\chi \eqn \top$, and $\varsigma$ is consistent with step $i$ of run $\rho$.
Then there is some $(\top,\varsigma') \in \delta(\inquotes{\psi})$
such that $\neg \last\not\in \varsigma'$ and $\varsigma'$ is consistent with step $i$ of $\rho$ as well.
\item[(3)] If $\chi$ is not $\top$ or $\bot$ then
$\varsigma$ has $\last$ or $\neg \last$.
\end{inparaenum}
\end{lemma}
\begin{proof}
All three statements are shown simultaneously by induction on $\psi$.
\begin{compactitem}
\item
If $\psi$ is $\top$, $\bot$, or an atom $K$ then $\chi$ is $\top$ or $\bot$, so (1) and (3) hold, and $\neg \last \not \in \varsigma$, so also (2) is satisfied.
\item
If $\psi = \X \psi'$ then 
$\delta(\inquotes{\psi}) = \{(\inquotes{\psi'},\{\neg \last\}), (\inquotes{\bot}, \{\last\})\}$.
(1) is satisfied by $(\inquotes{\bot}, \{\last\})$,
(2) holds because $\psi'$ cannot be $\top$ since $\top$ does not occur in $\LTLconf$, and
(3) is satisfied anyway.
\item
If $\psi = \Xw \psi'$ then 
$\delta(\inquotes{\psi}) = \{(\inquotes{\psi'},\{\neg \last\}), (\inquotes{\top}, \{\last\})\}$.
(1) is satisfied by $(\inquotes{\top}, \{\last\})$,
(2) holds because $\psi'$ cannot be $\top$ since $\top$ does not occur in $\LTLconf$, and
(3) is satisfied anyway.
\item
If $\psi = \psi_1 \vee \psi_2$ then we must have $\chi = \chi_1\vee \chi_2$
such that $(\chi_i,\varsigma_i) \in \delta(\inquotes{\psi_i})$ for both 
$i\in \{1,2\}$, and $\varsigma = \varsigma_1 \cup \varsigma_2$.

(1)
Suppose $\last \inn \varsigma$ and $\neg \last \not\in \varsigma$.
First, assume $\last \in \varsigma_1$, 
$\neg \last \not\in \varsigma_1$,
and $\neg \last \not\in \varsigma_2$.
By the induction hypothesis (1), $\chi_1$ is either $\top$ or $\bot$.
In the former case, $\chi = \top$, so the claim holds.
Otherwise, $\chi = \chi_2$.
Then, if $\last \in \varsigma_2$ we can again use the induction hypothesis
to conclude that $\chi = \chi_2$ is $\top$ or $\bot$.
Otherwise, we have $\last \not\in \varsigma_2$ and
$\neg \last \not\in \varsigma_2$, so
$\chi_2$ must be $\top$ or $\bot$
by the induction hypothesis (3).

(2)
Suppose $\neg \last \inn \varsigma$, $\last \not\in \varsigma$, and $\chi=\top$,
and $\varsigma$ is consistent with step $i$ of $\rho$.
W.l.o.g., we can assume $\neg \last \in \varsigma_1$, 
$\last \not\in \varsigma_1$, $\last \not\in \varsigma_2$,
and $\chi_1 = \top$.
By the induction hypothesis (2) applied to $\chi_1$, there is some
$(\top,\varsigma_1') \in \delta(\inquotes{\psi_1})$
such that $\neg \last \not\in\varsigma_1'$ and $\varsigma_1'$ is consistent with step $i$ of $\rho$.
By \lemref{delta total}, there is some 
$(\chi_2',\varsigma_2') \in \delta(\inquotes{\psi_2})$ such that
$\varsigma_2'$ is consistent with step $i$ of $\rho$, and such that $\neg \last \not\in \varsigma_2'$.
Thus $(\top,\varsigma') \in \delta(\inquotes{\psi})$
with $\varsigma' = \varsigma_1'\cup \varsigma_2'$ satisfies the claim.

(3)
If $\chi$ is not $\top$ or $\bot$ then at least one of $\chi_1$ or
$\chi_2$ is not $\top$ or $\bot$, so by the induction hypothesis (3),
$\varsigma_1$ or $\varsigma_2$ contains $\last$ or $\neg \last$, hence so
does $\varsigma$.
\item
If $\psi = \psi_1 \wedge \psi_2$ then we must have $\chi = \chi_1\wedge \chi_2$
such that $(\chi_i,\varsigma_i) \in \delta(\inquotes{\psi_i})$ for both 
$i\in \{1,2\}$.

(1)
Suppose $\last \in \varsigma$ and $\neg \last \not\in \varsigma$.
W.l.o.g., we can assume $\last \in \varsigma_1$, 
$\neg \last \not\in \varsigma_1$,
and $\neg \last \not\in \varsigma_2$.
By the induction hypothesis (1), $\chi_1$ is either $\top$ or $\bot$.
In the latter case, $\chi = \bot$, so the claim holds.
Otherwise, $\chi = \chi_2$, and as by assumption
$\neg \last \not\in \varsigma_2$, by the induction hypothesis (3), 
$\chi_2$ must be $\top$ or $\bot$.

(2)
Suppose $\neg \last \in \varsigma$, $\last \not\in \varsigma$, and $\chi=\top$,
and $\varsigma$ is consistent with step $i$ of $\rho$.
We can assume  $\last \not\in \varsigma_1$, $\last \not\in \varsigma_2$ and
$\chi_1 = \chi_2=\top$.
We must have $\last \in \varsigma_1$, $\last \in \varsigma_2$, or both.
However, for each $i\in \{1,2\}$ such that $\last \in \varsigma_i$, by
the induction hypothesis (2) there is some
$(\top,\varsigma_i') \in \delta(\inquotes{\psi_i})$
such that $\last \not\in\varsigma_i'$ and $\varsigma_i'$ is consistent with step $i$ of $\rho$.
If $\last \not\in \varsigma_i$, set $\varsigma_i' = \varsigma_i$.
Hence $(\top,\varsigma_1'\cup \varsigma_2') \in \delta(\inquotes{\psi})$ such that
$\varsigma_1'\cup \varsigma_2'$ is consistent with step $i$ of $\rho$ and 
$\neg \last\not\in \varsigma_1'\cup \varsigma_2'$.

(3)
If $\chi$ is not $\top$ or $\bot$ then at least one of $\chi_1$ or
$\chi_2$ is not $\top$ or $\bot$, so by the induction hypothesis (3),
$\varsigma_1$ or $\varsigma_2$ contains $\last$ or $\neg \last$, hence so
does $\varsigma$.
\item For $\psi = \G \psi'$, note that
$\delta_\lambda = \{(\inquotes{\top},\{ \lambda\}),(\inquotes{\bot},\{\neg\lambda\})\}$ satisfies the properties.
The result then follows from the cases for $\vee$ and $\wedge$.
\item All other cases follow from the cases for $\vee$ and $\wedge$.
\qed
\end{compactitem}
\end{proof}


\begin{numberedlemma}{\ref{lem:nfa}}
$\NFApsi$ accepts a word that is consistent with a run $\rho$ iff 
$\rho,0 \modelsLTL \psi$.
\end{numberedlemma}
\begin{proof}
($\Longrightarrow$)
Let $w = \varsigma_0 \varsigma_1 \cdots \varsigma_{n}$ be accepted, and
$q_0 \goto{\varsigma_0} q_1 \goto{\varsigma_1} \dots \goto{\varsigma_{n}} q_{n+1}$
be the respective accepting run of $\NFApsi$.
By \defref{NFA}, there are $\varsigma_i'$, 
such that $\varsigma_i = \varsigma_i' \setminus \{\last, \neg \last\}$
and $\{\last, \neg \last\} \not\subseteq \varsigma_i'$
for all $i$, $0\leq i \leq n$.
Let $w'$ be the word $w' = \varsigma_0' \varsigma_1' \cdots \varsigma_{n}'$.
Then $w'$ is consistent with $\rho$ because so is $w$.
Moreover,
by \lemref{delta last} (2) we can choose $\varsigma_{n}'$ such that
$\neg \last \not\in \varsigma_{n}'$, and $\varsigma_{n}'$ is consistent with $\rho$ at $n$.
Then $w'$ is well-formed: indeed,
since edges to $\inquotes{\top}$ labeled $\last$ are redirected to $q_e$
and $\inquotes{\bot}$ cannot occur in the accepting sequence, by \lemref{delta last} (1)
we have $\last\not\in \varsigma_i'$ for $i<n$.
Thus by \defref{NFA} we have
$\inquotes{\top} \in \delta^*(\inquotes{\psi},w')$.
According to \lemref{deltastar}, $\rho \modelsLTL \psi$.

($\Longleftarrow$)
If $\rho \modelsLTL \psi$ then by \lemref{deltastar} there is a 
well-formed word 
$w = \varsigma_0 \varsigma_1 \cdots \varsigma_{n}$ 
that is consistent with $\rho$ such that $\inquotes{\top} \in \delta^*(\inquotes{\psi},w)$.
As $w$ is well-formed, no $\varsigma_i$ has both $\last$
and $\neg \last$. Hence all $\delta$-steps are reflected by transitions in 
$\NFApsi$. As $\inquotes{\psi}$ is the initial state,
by \defref{NFA}, there is an accepting run
in $\NFApsi$ to $\inquotes{\top}$ or $q_e$.
\qed
\end{proof}

\end{document}